\documentclass[journal]{IEEEtran}

\usepackage{amsthm}
\usepackage[cmex10]{amsmath}
\usepackage{graphicx}
\usepackage{color}
\usepackage[tight,footnotesize]{subfigure}
\usepackage{amsfonts}
\usepackage{algorithmic}
\usepackage{algorithm}
\usepackage{graphicx}
\usepackage{subfigure}
\usepackage{bm}
\usepackage{times,color,amssymb,epsfig,psfrag,stfloats,enumerate}
\usepackage{cite}

\newtheorem{theorem}{Theorem}

\newtheorem{lemma}{Lemma}
\newtheorem{corollary}{Corollary}
\newtheorem{proposition}{Proposition}

\ifCLASSINFOpdf
\else
\fi
\begin{document}
%
\title{Achievable Rate Region of Non-Orthogonal Multiple Access Systems with Wireless Powered Decoder}

\author{Jie~Gong~\IEEEmembership{Member,~IEEE}, Xiang~Chen~\IEEEmembership{Member,~IEEE}
\thanks{The authors are with Guangdong Key Laboratory of Information Security Technology, Sun Yat-sen University, Guangzhou 510006, Guangdong, China. Emails: \{gongj26, chenxiang\}@mail.sysu.edu.cn. Xiang Chen, as the corresponding author, is also with Key Lab of EDA, Research Institute of Tsinghua University in Shenzhen, Shenzhen 518057, Guangdong, China, and SYSU-CMU Shunde International Joint Research Institute (JRI).}
\thanks{The work is supported in part by the NSFC (No.~61501527), the EU's Horizon 2020 Research and Innovation Staff Exchange programme (No.~734325, TESTBED), the MOE-CMCC Joint Research Fund of China (MCM20160104), State's Key Project of Research and Development Plan (No.~2016YFE0122900-3), the Fundamental Research Funds for the Central Universities, Science, Technology and Innovation Commission of Shenzhen Municipality (No.~JCYJ20150630153033410), SYSU-CMU Shunde International Joint Research Institute, Guangdong Science and Technology Project (No.~2016B010126003), and 2016 Major Project of Collaborative Innovation in Guangzhou.}
}

\maketitle

\begin{abstract}
Non-orthogonal multiple access (NOMA) is a candidate multiple access scheme in 5G systems to simultaneously accommodate tremendous number of wireless nodes. On the other hand, RF-enabled wireless energy harvesting is a promising technology for self-sustained wireless devices. In this paper, we study a NOMA system where the near user harvests energy from the strong radio signal to power the information decoder. {Both constant and dynamic decoding power consumption models are considered. For the constant decoding power model, the achievable rate regions for time switching and power splitting are characterized in closed-form. A generalized scheme is proposed by combining the conventional time switching and power splitting schemes, and its achievable rate region can be found by solving two convex optimization subproblems.} For the dynamic decoding power model where the decoding power consumption is proportional to data rate, the achievable rate region can be found by a low-complexity search algorithm. Numerical results show that the achievable rate region of the generalized scheme is larger than those of the time switching scheme and power splitting scheme, and rate-dependent decoder design helps to enlarge the achievable rate region.
\end{abstract}

\begin{IEEEkeywords}
Non-orthogonal multiple access, successive interference cancellation, energy harvesting, wireless power transfer, decoding power consumption.
\end{IEEEkeywords}

\IEEEpeerreviewmaketitle

\renewcommand{\algorithmicrequire}{\textbf{Input:}}
\renewcommand{\algorithmicensure}{\textbf{Output:}}

\section{Introduction}
With the tremendous increase of mobile users for high-speed video services and machine-type communication nodes for Internet-of-Things (IoT) applications in the coming years, the next generation mobile communication systems (5G) will encounter severe scarcity of radio resources. Non-orthogonal multiple access (NOMA), in which the users access the network in the same frequency band and are distinguished in the power domain, is expected as one of the candidate solutions to greatly enhance the spectral efficiency \cite{dai2015NOMA}. The basic idea of NOMA (as proposed in \cite{saito2013non}) is that two users, a far user and a near user, access in the same frequency band, and are allocated with higher and lower transmit power, respectively. The far user directly decodes its desired information, while the near user adopts successive interference cancellation (SIC) to firstly decode and subtract the signal of the far user, and then decode the desired information for itself. Inspired by the fact that the near user receives strong radio signal (including desired signal and interference), it is promising to apply RF-enabled wireless energy harvesting technology \cite{krikidis2014simultaneous, ding2015application} for this near user, which facilitates the application of self-sustained devices in the IoT networks \cite{bi2016wireless}, such as sensor nodes and RFIDs.

NOMA has been extensively studied in the literature. Saito, et. al \cite{saito2013non} firstly proposed the concept for 5G, and system-level performance was evaluated in \cite{saito2013system}, which shows that both network capacity and  cell-edge user throughput can be improved compared with the conventional orthogonal multiple access. Uplink NOMA was considered in \cite{al2014uplink}, in which a novel subcarrier and power allocation algorithm was proposed to maximize the users' sum-rate. Taking user fairness into account, the power allocation problem was studied in \cite{timotheou2015fairness}. In multiuser scenario, user pairing was introduced, and its impact on the performance was characterized in \cite{ding2016impact}. Further, since the near user can decode the information of the far user, the authors in \cite{ding2015cooperative} proposed a cooperative NOMA scheme, i.e., the near user can serve as a relay to help the far user. The work was further extended to combine simultaneous wireless information and power transfer (SWIPT) \cite{liu2016cooperative, xu2017joint}. However, the previous work only considers transmit power consumption. In fact, the receiving circuit power consumption cannot be ignored in applications such as wireless sensors \cite{cui2005energy}. In NOMA system, SIC receiver in the near user consumes more energy than conventional receivers. A constant receive power model with ambient energy source is adopted in \cite{zhou2015outage}, which provides an intuition on how to efficiently allocate the harvested energy. Nevertheless, as wireless radio energy is controllable compared with ambient energy, how to split the received signal between energy harvesting and information decoding is an open problem.

{For RF-enabled wireless energy harvesting, two types of receiver structures, power splitting receiver and time switching receiver, have been initially proposed in \cite{zhang2013mimo} and extensively studied in \cite{liu2013wireless} and \cite{zhou2013wireless}.} The research on SWIPT further extends to multiuser OFDM systems \cite{ng2013wireless, zhou2014wireless}, relay networks \cite{zhu2015wireless}, and cellular network \cite{huang2014enabling}. Since the receiver can harvest energy from both desired signal and interference, strong interference will benefit wireless energy harvesting. Based on this, two-user interference channel and multi-user interference channel were studied in \cite{park2013joint} and \cite{lee2015collaborative}, respectively. However, most of the previous works mainly focus on characterizing the rate-energy tradeoff by adjusting the portion of received signal on energy harvesting and information decoding. Again, the harvested energy is mainly used for data transmission rather than receiving.

In this paper, we study a NOMA system composed of a base station (BS) and two users, where the near user is powered by wireless energy harvesting. It splits the received radio signal into two parts for energy harvesting and information decoding, respectively. The harvested energy is used to power the information decoding module. We aim to characterize the achievable rate region $(R_1, R_2)$ of the near user (with rate $R_1$) and the far user (with rate $R_2$). A generalized energy harvesting scheme based on the conventional time switching and power splitting is proposed. To obtain the achievable rate region, the problem of maximizing the achievable rate of UE 2 given the target rate of UE 1 is formulated and solved under both constant and dynamic decoding power assumptions. With constant decoding power consumption, two special cases, time switching scheme and power splitting scheme, are firstly studied and closed-form solutions are obtained. Then the analysis is extended to the generalized scheme, where the rate maximization is a convex optimization problem, and hence can be efficiently solved by a numerical search algorithm. Numerical results show that the generalized scheme achieves a larger rate region compared with the conventional time switching and power splitting schemes. With dynamic decoding power consumption, both optimal and suboptimal algorithms are proposed, and the achievable rate regions with different parameter settings are analyzed.

The primary contributions of this paper are summarized as follows.
\begin{itemize}
\item {To the best of our knowledge, this is the first time considering decoding power consumption in wireless powered NOMA system. To characterize the achievable rate region, a rate maximization problem is formulated by maximizing the achievable rate of UE 2 under the constraint of UE 1's target rate.}
\item With constant decoding power model, the boundaries of the achievable rate regions for time switching scheme and power splitting scheme are characterized in closed-form. For time switching scheme, there exists a \emph{cut-off} line on the rate region. The rate pairs on the right side of this cut-off line are not achievable. For power splitting scheme, there is a stringent feasibility requirement on the channel quality of UE 1.
\item {A generalized energy harvesting scheme is proposed, for which the boundary of the achievable rate region with constant decoding power consumption can be obtained by solving two convex optimization subproblems.} It is proved that the optimal solution is the most \emph{economical} one, i.e., the harvested energy is equal to the information decoding power consumption. It is also shown that the generalized scheme has a larger rate region compared with the conventional schemes.
\item With dynamic decoding power model, an exhaustive search algorithm and an efficient suboptimal algorithm is proposed to find the boundary of the achievable rate region. Numerical results show that by trading rate with power consumption, the rate region can be greatly extended compared with that under constant power model.
\end{itemize}

The rest of the paper is organized as follows. Section \ref{sec:model} describes the system model and the problem formulation. The problems under constant and dynamic decoding power assumption are analyzed in Sections \ref{sec:const} and \ref{sec:dyn}, respectively. Extensive results and discussions are presented in Section \ref{sec:comp}. Finally, Section \ref{sec:concl} concludes the paper.

\section{System Model and Problem Formulation} \label{sec:model}
Consider a downlink wireless communication system as shown in Fig.~\ref{fig:system}, where a BS serves a near user (UE 1) and a far user (UE 2) in the same frequency band. Power domain NOMA scheme is applied in this system, i.e., the BS simultaneously transmits data to both users by superposition coding with different power levels. After receiving the coded signal, UE 2 directly decodes its desired information by viewing the interference from UE 1 as noise. UE 1 applies SIC scheme to improve its performance by firstly decoding and removing the signal for UE 2, and then decoding the desired signal for itself. By adjusting the power allocation between two users, per-user data rate can be flexibly controlled. UE 1 is also equipped with an energy harvester, in which the energy carried in the radio signal is harvested to power the signal receiving and decoding module. {As the received signal strength is weak, UE 2 is powered in other ways such as battery. The model can be extended to multi-user case by grouping users into pairs, each of which consists of a near user and a far user. The user pairs are allocated with orthogonal frequency bands, so that each pair can be managed independently. While how to group users into pairs can refer to \cite{ding2016impact}.}

\begin{figure}
\centering
\includegraphics[width=3.0in]{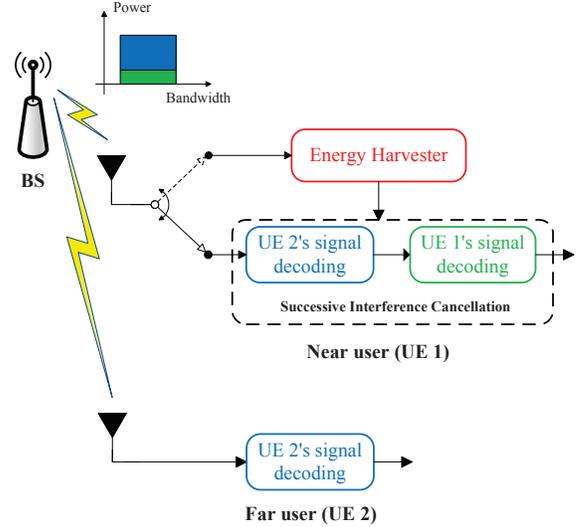}
\caption{A NOMA system with wireless powered SIC decoder in the near user.} \label{fig:system}
\end{figure}

\subsection{Signal Model}
The desired signal for UE $i \in \{1, 2\}$ is denoted by $x_i$ with mean $\mathbb{E}[|x_i|^2] = 1$. The superposition coded signal can be represented as
\begin{align}
x = \sqrt{P_1}x_1 + \sqrt{P_2}x_2,
\end{align}
where $P_i$ is the transmit power for UE $i$. The received signal at UE $i$ is
\begin{align}
y_i &= h_ix + n_i, \nonumber\\
&= h_i\sqrt{P_1}x_1 + h_i\sqrt{P_2}x_2 + n_i, \quad i = 1, 2,
\end{align}
where $h_i$ is the channel coefficient from the BS to UE $i$, $n_i$ is the additive white Gaussian noise (AWGN) with zero mean and variance $\sigma^2$. Assume that ${|h_1|^2} > {|h_2|^2}$, and the link quality $|h_1|^2$ is good enough so that UE 1 can be powered by wireless energy harvesting.

{It is remarkable that only transmit power adaptation is considered in this paper. Since the end-to-end energy transfer efficiency is low in current state-of-art \cite{zeng2017commun}, the BS is most likely equipped with multi-antenna to achieve highly directional power transmission, and hence, transmit beamforming needs to be considered. This paper can be viewed as a special case with fixed beamformer. Joint optimization of transmit power and beamformer for multi-antenna case is left for future work.
}

\subsection{Energy Harvesting Scheme and Constraints}
{In general, the conventional power splitting scheme may not work as the received power may be less than the power consumption of the information decoding module.} In this paper, a generalized scheme for energy harvester design is proposed, which combines time switching scheme and power splitting scheme. Specifically, each transmission slot is divided into two sub-slots. In the first sub-slot, UE 1 works as a \emph{pure} energy harvester and only the information for UE 2 is carried in the radio signal. Denote $P_2^{(1)}$ as the transmit power for UE 2 in the first sub-slot. {The BS's transmit power is subject to a peak power constraint, i.e.,}
\begin{align}
P_2^{(1)} \le P_{\mathrm{max}}, \label{eq:Pmax_1}
\end{align}
where $P_{\mathrm{max}}$ is the BS's power budget. According to information theory, the achievable rate of UE 2 in the first sub-slot can not exceed the channel capacity, i.e.,
\begin{align}
R_2^{(1)} \le \log_2 \bigg( 1+ \frac{|h_2|^2 P_2^{(1)}}{\sigma^2}\bigg). \label{eq:R2_1}
\end{align}

In the second sub-slot, UE 1 adopts power splitting scheme, i.e., a portion of signal power, denoted by $\rho (0 \le \rho \le 1)$, is split to energy harvester, and the other portion, denoted by $1-\rho$, is split to information decoder. Denote $P_1^{(2)}$ and $P_2^{(2)}$ as the power allocation for the respective users in the second sub-slot. Under the peak power constraint, we have
\begin{align}
P_1^{(2)} + P_2^{(2)} \le P_{\mathrm{max}}. \label{eq:Pmax_2}
\end{align}

Since the interference can be completely eliminated by SIC, UE 1 can be viewed as an interference-free receiver. Hence, its achievable rate satisfies
\begin{align}
R_1^{(2)} \le \log_2 \bigg( 1 + \frac{|h_1|^2 (1-\rho) P_1^{(2)}}{\sigma^2}\bigg). \label{eq:R1_2}
\end{align}

UE 2 decodes its desired signal by viewing the interference from UE 1 as noise. The achievable rate is constrained by
\begin{align}
R_2^{(2)} \le \log_2 \bigg( 1 + \frac{|h_2|^2P_2^{(2)}}{|h_2|^2P_1^{(2)}+\sigma^2}\bigg). \label{eq:R2_2}
\end{align}
In addition, to guarantee that SIC works, UE 1 should be able to successfully decode the data for UE 2. Thus, besides (\ref{eq:R2_2}), $R_2^{(2)}$ is also constrained by the decoding ability of UE 1, i.e.,
\begin{align}
R_2^{(2)} \le \log_2 \bigg( 1 + \frac{|h_1|^2 (1-\rho) P_2^{(2)}}{|h_1|^2 (1-\rho) P_1^{(2)}+\sigma^2}\bigg). \label{eq:R2_2sic}
\end{align}

Without loss of generality, assume the slot length is normalized to 1. 
Denote the length of the first sub-slot by $t (0 \le t \le 1)$, and that of the second sub-slot by $1-t$. The total harvested energy should be no smaller than the required energy for information decoding, i.e.,
\begin{align}
t\xi |h_1|^2  P_2^{(1)} \!+\! (1 \!-\! t) \rho \xi |h_1|^2 \big(P_1^{(2)} \!+\! P_2^{(2)}\big) \ge (1 \!-\! t)P_{\mathrm{SIC}}, \label{eq:Psic}
\end{align}
where $0 < \xi < 1$ is the energy harvesting efficiency, and $P_{\mathrm{SIC}}$ is the total power consumption for information decoding. {In this paper, a linear energy harvesting model ($\xi$ is constant) is adopted for analytical tractability. This model has been widely used in the literature. Although in practice, the relationship between received power and harvested power maybe non-linear \cite{boshkovska2015practical}, $\xi$ is approximately constant at low input power regime.}

The conventional time switching scheme can be viewed as a special case of the generalized scheme with $\rho = 0$, and the conventional power splitting scheme is a special case with $t = 0$. Later, we will extensively study generalized scheme as well as the two special cases.

\subsection{Problem Formulation}
In this paper, we aim to characterize the achievable rate regions of UE 1 and UE 2 under the proposed schemes, which can be represented as
\begin{align}
\mathcal{R} = \{&(R_1, R_2) \big| R_1 = (1-t)R_1^{(2)}, \nonumber\\
{}& R_2 = t R_2^{(1)}+(1-t)R_2^{(2)}, \nonumber\\
{}& 0 \le t \le 1, 0 \le \rho \le 1, \nonumber\\
{}& R_2^{(1)}\!, P_2^{(1)}\!, R_1^{(2)}\!,  R_2^{(2)}\!, P_1^{(2)}\!, P_2^{(2)} \textrm{~satisfy~(\ref{eq:Pmax_1})\textrm{-}(\ref{eq:Psic})} \}, \label{eq:region}
\end{align}
where all the variables are nonnegative. The achievable rate region can be described by the boundary curve which maps each feasible $R_1$ to the maximum of $R_2$. Thus, the boundary curve of the achievable rate region (\ref{eq:region}) can be found by maximizing $R_2$ for each $R_1 = r$, which is formulated as
\begin{subequations}\label{prob:mix}
\begin{align}
(\mathbf{P0}) \quad \max_{t, \rho, \bm{R}, \bm{P}} \;&\; t R_2^{(1)} + (1-t)R_2^{(2)} \label{prob:mixobj}\\
\mathrm{s.t.} \;&\; (1-t)R_1^{(2)} = r, \label{prob:R1cnstr}\\
\;&\; 0\le t < 1, 0 \le \rho < 1, \label{prob:mixrr}\\
\;&\; (\ref{eq:Pmax_1})\textrm{-}(\ref{eq:Psic}), \nonumber
\end{align}
\end{subequations}
where the optimization variables include $t$, $\rho$, $\bm{R} = [R_2^{(1)}, R_1^{(2)}, R_2^{(2)}]$ and $\bm{P} = [P_2^{(1)}, P_1^{(2)}, P_2^{(2)}]$. In addition, as $r=0$ is trivial, only the case $r>0$ needs to be studied. In this case, it can be easily found that $t \neq 1$ and $\rho \neq 1$ according to (\ref{prob:R1cnstr}). Thus, we have $t < 1$ and $\rho < 1$ as in (\ref{prob:mixrr}) so that the terms $\frac{1}{1-t}$ and $\frac{1}{1-\rho}$ which will appear in the following contents are well defined.

{Notice that according to (\ref{eq:Pmax_2}) and (\ref{eq:R1_2}), the feasible range of $r$ can be expressed as $0 \le r \le (1-t)\log_2 \big( 1 + \frac{|h_1|^2 (1-\rho) P_{\max}}{\sigma^2}\big)$, where $t$ and $\rho$ satisfy (\ref{eq:Psic}). There is no closed-form feasible range of $r$ in general, but for some special cases, the range can be explicitly given as detailed in the next section.}

In reality, the information decoding power consumption $P_{\mathrm{SIC}}$ can be modeled either constant or dynamic based on receiver circuit structure and decoding scheme. In the following sections, we will solve the problem ($\mathbf{P0}$) under different power consumption models.

\section{Constant Decoding Power Consumption Case} \label{sec:const}
In this section, we characterize the achievable rate region with constant power consumption model. This model has been widely used in the literature {(see \cite{wang2006realistic, imran2011energy} for example)} and well models the iterative decoding scheme with fixed number of iterations in real system. In addition, well-structured results can be obtained based on this model as shown later. These results are helpful for a better understanding of the relation between parameters and performance, and provide some intuitions for the design of real systems.

At first, with constant decoding power consumption, some quick observations on the problem ($\mathbf{P0}$) are summarized as follows.
\begin{lemma} \label{lemma:equal}
\textbf{(Equality Constraints)} To achieve the maximum of the problem ($\mathbf{P0}$), the constraints (\ref{eq:Pmax_1}), (\ref{eq:R2_1}), (\ref{eq:Pmax_2}), (\ref{eq:R1_2}) and (\ref{eq:R2_2}) (or (\ref{eq:R2_2sic})) are satisfied with equality.
\end{lemma}
\begin{proof}
See Appendix \ref{proof:equal}.
\end{proof}

{Based on Lemma \ref{lemma:equal}, all the optimization variables can be represented in terms of $t$ and $\rho$. Specifically, as (\ref{eq:Pmax_1}) and (\ref{eq:R2_1}) are tightly satisfied, $P_2^{(1)}$ and $R_2^{(1)}$ become constant. $R_1^{(2)}$ can be directly written in terms of $t$ and $\rho$ according to (\ref{prob:R1cnstr}). As (\ref{eq:R1_2}) is satisfied with equality, $P_1^{(2)}$ can be written as a function of $R_1^{(2)}$ and $\rho$. With equality in (\ref{eq:Pmax_2}), we have $P_2^{(2)} = P_{\max} - P_1^{(2)}$. Finally, $R_2^{(2)}$ can also be written in terms of $t$ and $\rho$ according to either (\ref{eq:R2_2}) or (\ref{eq:R2_2sic}) based on the following lemma.}

\begin{lemma} \label{lemma:R2}
\textbf{(Constraint on $R_2^{(2)}$)} If
\begin{align}
|h_2|^2 \le (1-\rho) |h_1|^2, \label{eq:constrR2}
\end{align}
$R_2^{(2)}$ is constrained by (\ref{eq:R2_2}). Otherwise, it is constrained by (\ref{eq:R2_2sic}).
\end{lemma}
\begin{proof}
The lemma is proved by directly comparing the right hand sides of (\ref{eq:R2_2}) and (\ref{eq:R2_2sic}).
\end{proof}

To start with, two special cases are considered: time switching scheme and power splitting scheme, where there is only a single variable $t$ or $\rho$. Then the analysis is extended to the generalized scheme.

\subsection{Time Switching Scheme} \label{sec:ts}
Recall that time switching scheme is a special case of the generalized scheme with $\rho = 0$. In this case, (\ref{eq:constrR2}) always holds and hence, $R_2$ is constrained by (\ref{eq:R2_2}), and the power constraint (\ref{eq:Psic}) degrades to
\begin{align}
t\xi |h_1|^2  P_{\mathrm{max}} \ge (1-t)P_{\mathrm{SIC}}. \label{eq:Psicts}
\end{align}
As the constraints (\ref{eq:Pmax_1})-(\ref{eq:R2_2}) are all satisfied with equality, the problem ($\mathbf{P0}$) can be reformulated with a single variable $t$ as
\begin{subequations}\label{prob:tsequal}
\begin{align}
(\mathbf{P1}) \;\: \max_t &\; f_0(t) \\
\mathrm{s.t.} &\; \frac{P_{\mathrm{SIC}}}{\xi |h_1|^2 P_{\mathrm{max}} \!+\! P_{\mathrm{SIC}}} \!\le\! t \!\le\! 1\!-\!\frac{r} {\log_2\big( 1\!+ \!\frac{|h_1|^2P_{\mathrm{max}}}{\sigma^2}\big)}, \label{prob:tstf}
\end{align}
\end{subequations}
where the objective function is
\begin{align}
f_0(t) = &\: \log_2 \bigg( 1+ \frac{|h_2|^2P_{\mathrm{max}}}{\sigma^2}\bigg) - \nonumber\\
&\: (1-t) \log_2 \bigg( \frac{|h_2|^2}{|h_1|^2} \big(2^{\frac{r}{1-t}}-1\big) + 1\bigg),
\end{align}
the left hand side of (\ref{prob:tstf}) is derived from (\ref{eq:Psicts}), and the right hand side of (\ref{prob:tstf}) holds as $r = \le (1-t)\log_2 \big(1+ \frac{|h_1|^2P_{\mathrm{max}}}{\sigma^2} \big)$. We have the following lemma.
\begin{lemma} \label{lemma:mono}
\textbf{(Monotonicity)} The function $f_0(t)$ is non-increasing for $0 \le t < 1$.
\end{lemma}
\begin{proof}
See Appendix \ref{proof:mono}.
\end{proof}

Based on Lemma \ref{lemma:mono}, the optimal solution of problem ($\mathbf{P1}$) can be directly obtained by taking the minimum value of $t$. The result is summarized in the following theorem.
\begin{theorem} \label{thm:ts}
The optimal solution for the problem ($\mathbf{P1}$) is
\begin{align}
{}& R_{2,\mathrm{TS}}^* = \log_2 \bigg( 1+ \frac{|h_2|^2P_{\mathrm{max}}}{\sigma^2}\bigg) - \nonumber\\
{}& \frac{\xi |h_1|^2 P_{\mathrm{max}}}{\xi |h_1|^2 P_{\mathrm{max}} \!+\! P_{\mathrm{SIC}}} \log_2 \bigg( \frac{|h_2|^2}{|h_1|^2} \Big(2^{\frac{r(\xi |h_1|^2 P_{\mathrm{max}} \!+\! P_{\mathrm{SIC}})}{\xi |h_1|^2 P_{\mathrm{max}}}}\!-\!1\Big) \!+\! 1 \bigg), \label{eq:R2opt}
\end{align}
which is achieved when
\begin{align}
t = \frac{P_{\mathrm{SIC}}}{\xi |h_1|^2 P_{\mathrm{max}} + P_{\mathrm{SIC}}}. \label{eq:tsoptcon}
\end{align}
The feasible range of $r$ is
\begin{align}
0 \le r \le \frac{\xi |h_1|^2 P_{\mathrm{max}}}{\xi |h_1|^2 P_{\mathrm{max}} + P_{\mathrm{SIC}}} \log_2 \bigg( 1+ \frac{|h_1|^2P_{\mathrm{max}}}{\sigma^2}\bigg). \label{eq:R1feasi}
\end{align}
\end{theorem}
In Theorem \ref{thm:ts}, the feasible range of $r$ is obtained by guaranteeing that the feasible range of $t$ defined by (\ref{prob:tstf}) is nonempty. Notice that when (\ref{eq:tsoptcon}) holds, (\ref{eq:Psicts}) is satisfied with equality. That is, Theorem \ref{thm:ts} holds as the amount of harvested energy equals to the required energy for information decoding, which is intuitively the most \emph{economical} way of dividing the received signal between energy harvesting and information decoding.

Based on Theorem \ref{thm:ts}, we have the following corollary.
\begin{corollary} \label{cor:ps}
For all $r$ satisfying (\ref{eq:R1feasi}), we have
\begin{align}
\frac{P_{\mathrm{SIC}}}{\xi |h_1|^2 P_{\mathrm{max}} + P_{\mathrm{SIC}}}R_{2,\mathrm{max}} \le R_{2,\mathrm{TS}}^* \le R_{2,\mathrm{max}},
\end{align}
where $R_{2,\mathrm{max}} = \log_2 \big( 1+ \frac{|h_2|^2P_{\mathrm{max}}}{\sigma^2}\big)$.
\end{corollary}
Corollary \ref{cor:ps} tells us that when $R_1 = 0$, $R_{2,\mathrm{TS}}^*$ achieves its maximum, i.e., all the power resource is allocated to UE 2. However, when $R_1$ equals to the right hand side of (\ref{eq:R1feasi}), $R_{2,\mathrm{TS}}^*$ achieves its minimum, which is strictly larger than 0. It means that there is a \emph{cut-off} line on achievable rate region for time switching scheme. When $R_1$ achieves the maximal value, $R_1$ cannot be further increased by sacrificing $R_2$. This is due to the nature of time switching: to harvest energy, the length of the first sub-slot can never be zero.

\subsection{Power Splitting Scheme}
In power splitting scheme, we have $t = 0$. In this case, the power constraint (\ref{eq:Psic}) degrades to
\begin{align}
\rho \xi |h_1|^2 P_{\mathrm{max}} \ge P_{\mathrm{SIC}}. \label{eq:Psicps}
\end{align}
Based on the above equation, there is a feasibility issue for power splitting scheme, as detailed in the following lemma.
\begin{lemma} \label{lemma:feasible}
\textbf{(Feasibility)} If
\begin{align}
|h_1|^2 \le \frac{P_{\mathrm{SIC}}}{\xi P_{\mathrm{max}}}, \label{eq:deepfade}
\end{align}
there does not exist any achievable rate pair $(R_1, R_2)$ in power splitting mode so that $R_1 > 0$ and $R_2 > 0$.
\end{lemma}

Lemma \ref{lemma:feasible} is directly obtained from (\ref{eq:Psicps}) as if (\ref{eq:deepfade}) holds, (\ref{eq:Psicps}) cannot be satisfied with any feasible value $0 \le \rho < 1$. According to the above lemma, power splitting scheme is not applicable if the channel of near user is in deep fading.

If $|h_1|^2 > \frac{P_{\mathrm{SIC}}}{\xi P_{\mathrm{max}}}$, the problem ($\mathbf{P0}$) can be reformulated as
\begin{subequations}\label{prob:ps}
\begin{align}
(\mathbf{P2}) \quad \max_{\rho} \;&\; \min\bigg\{\log_2 \bigg(\frac{|h_2|^2P_{\mathrm{max}} + \sigma^2}{|h_2|^2P_1^{(2)} + \sigma^2}\bigg), \nonumber\\
\;&\; \qquad \log_2 \bigg(\frac{|h_1|^2 (1-\rho) P_{\mathrm{max}} + \sigma^2}{|h_1|^2 (1-\rho) P_1^{(2)}+\sigma^2}\bigg) \bigg\} \label{prob:psobj} \\
\mathrm{s.t.} \;&\; \log_2 \bigg(1+\frac{|h_1|^2(1-\rho)P_1^{(2)}}{\sigma^2} \bigg) = r, \label{prob:psP1} \\
\;&\; \frac{P_{\mathrm{SIC}}}{\xi |h_1|^2 P_{\mathrm{max}}} \le \rho < 1,
\end{align}
\end{subequations}
where $P_1^{(2)}$ in the objective function can actually be represented in terms of $\rho$ based on (\ref{prob:psP1}), and the solution is summarized in the following theorem.
\begin{theorem} \label{thm:ps}
If $|h_1|^2 > \frac{P_{\mathrm{SIC}}}{\xi P_{\mathrm{max}}}$, the optimal solution for the problem ($\mathbf{P2}$) is
\begin{align}
R_{2,\mathrm{PS}}^* = \left\{ \begin{array}{l} \log_2 \Big( \frac{(|h_2|^2 P_{\mathrm{max}} + \sigma^2) (\xi |h_1|^2 P_{\mathrm{max}} - P_{\mathrm{SIC}})}{(\xi P_{\mathrm{max}} (|h_2|^2(2^r-1) + |h_1|^2) - P_{\mathrm{SIC}}) \sigma^2} \Big), \\ \qquad\qquad\quad\quad \textrm{if~} |h_2|^2 \le |h_1|^2 - \frac{P_{\mathrm{SIC}}}{\xi P_{\mathrm{max}}}\\
\log_2 \Big( 1+ \frac{\xi|h_1|^2P_{\mathrm{max}} - P_{\mathrm{SIC}}}{\xi\sigma^2} \Big) - r, \quad \textrm{else}
\end{array}\right. \label{eq:optR2}
\end{align}
which is achieved when
\begin{align}
\rho = \frac{P_{\mathrm{SIC}}}{\xi |h_1|^2 P_{\mathrm{max}}}.
\end{align}
The feasible range of $r$ is
\begin{align}
0 \le r \le \log_2 \bigg( 1+ \frac{\xi|h_1|^2P_{\mathrm{max}} - P_{\mathrm{SIC}}}{\xi\sigma^2} \bigg). \label{eq:feasir}
\end{align}
\end{theorem}
\begin{proof}
As the objective function in (\ref{prob:psobj}) is a non-increasing function of $P_1^{(2)}$ and $\rho$, and $P_1^{(2)}$ is an increasing function of $\rho$, the maximum rate is achieved when $\rho = \frac{P_{\mathrm{SIC}}}{\xi |h_1|^2 P_{\mathrm{max}}}$ and $P_1^{(2)} = \frac{\xi P_{\mathrm{max}} (2^r-1)\sigma^2}{\xi |h_1|^2 P_{\mathrm{max}} - P_{\mathrm{SIC}}}$. Then the closed-form rate $R_{2,\mathrm{PS}}^*$ is obtained by comparing the two items in the minimization in (\ref{prob:psobj}).

The feasible range of $r$ is obtained according to (\ref{prob:psP1}).
\end{proof}

Notice that the optimal is achieved when the harvested energy is equal to the required decoding energy, which is the same with the time switching case. However, the difference is that the power splitting can tradeoff $R_1$ and $R_2$ completely. That is, when $R_1$ gradually decreases to zero, $R_2$ gradually increases to its maximum, and vice versa. In particular, when $|h_2|^2 > |h_1|^2 - \frac{P_{\mathrm{SIC}}}{\xi P_{\mathrm{max}}}$, $R_1$ trades off with $R_2$ linearly. However, power splitting scheme has stringent requirement on the channel gain of UE 1 (see Lemma \ref{lemma:feasible}). As a result, it is not applicable when the channel of UE 1 is not good enough.

\subsection{Generalized Scheme}
For the generalized scheme, based on Lemma \ref{lemma:R2}, the original problem ($\mathbf{P0}$) is divided into two subproblems, as detailed in the following subsections, respectively.

\subsubsection{Subproblem with $|h_2|^2 \le (1-\rho) |h_1|^2$} \label{sec:sub1}
When $|h_2|^2 \le (1-\rho) |h_1|^2$, $R_2^{(2)}$ is equal to the right hand side of (\ref{eq:R2_2}). Thus, the problem can be reformulated as
\begin{subequations}\label{prob:mix1}
\begin{align}
(\mathbf{P3.1}) \quad \max_{t, \rho} \;&\; \log_2 \bigg( 1+ \frac{|h_2|^2 P_{\mathrm{max}}}{\sigma^2}\bigg) - \nonumber\\
&\;(1-t) \log_2 \bigg( \frac{|h_2|^2}{(1-\rho) |h_1|^2} \big(2^{\frac{r}{1-t}}-1\big) + 1\bigg) \label{prob:mix1obj}  \\
\mathrm{s.t.} \;&\; 1-\rho \ge \frac{|h_2|^2}{|h_1|^2}, \label{prob:mix1rho}\\
{} \;&\; (1-t) \log_2 \bigg( 1 + \frac{|h_1|^2 (1-\rho) P_{\mathrm{max}}}{\sigma^2}\bigg) \ge r, \label{prob:mix1r}\\ 
{} \;&\; \frac{t}{1-t} + \rho \ge \frac{P_{\mathrm{SIC}}}{\xi |h_1|^2 P_{\mathrm{max}}}, \label{prob:mix1Psic}\\
{} \;&\; 0 \le t < 1, 0 \le \rho < 1, \label{prob:mix1trho}
\end{align}
\end{subequations}
where (\ref{prob:mix1r}) comes from the fact that $P_1^{(2)} \le P_{\mathrm{max}}$, and (\ref{prob:mix1Psic}) is equivalent to (\ref{eq:Psic}). As the objective function in (\ref{prob:mix1obj}) is a decreasing function of $\rho$, it is maximized when $\rho$ achieves its minimum, i.e., $\rho = \max \big\{0, \frac{P_{\mathrm{SIC}}}{\xi |h_1|^2 P_{\mathrm{max}}} - \frac{t}{1-t} \big\}$. Based on the optimal value of $\rho$, there are two cases as follows.

\textbf{Case (1.1):} If $0 \le \frac{P_{\mathrm{SIC}}}{\xi |h_1|^2 P_{\mathrm{max}}} - \frac{t}{1-t}$, i.e., $t \le t_U$, where
\begin{align}
t_U = \frac{P_{\mathrm{SIC}}}{\xi |h_1|^2 P_{\mathrm{max}} + P_{\mathrm{SIC}}},
\end{align}
we have $\rho = \frac{P_{\mathrm{SIC}}}{\xi |h_1|^2 P_{\mathrm{max}}} - \frac{t}{1-t}$. In this case, (\ref{prob:mix1Psic}) is satisfied with equality, and the problem can be transformed into
\begin{subequations}\label{prob:mix1t}
\begin{align}
(\mathbf{P3.1c}) \quad \max_{t} \;&\; f_1(t) \label{prob:mix1tobj}  \\
\mathrm{s.t.} \;&\; f_2(t) \ge r, \label{prob:mix1tr}\\
{} \;&\; \max \{0, t_L\} \le t \le t_U, \label{prob:mix1ttrho}
\end{align}
\end{subequations}
with a single variable $t$, where the objective function is
\begin{align}
f_1(t) =&\: \log_2 \bigg( 1+ \frac{|h_2|^2 P_{\mathrm{max}}}{\sigma^2}\bigg) - \nonumber\\
&\: (1\!-\!t) \log_2 \bigg( \frac{|h_2|^2}{|h_1|^2(\frac{1}{1\!-\!t} \!-\! \frac{P_{\mathrm{SIC}}}{\xi |h_1|^2 P_{\mathrm{max}}})} \big(2^{\frac{r}{1\!-\!t}}\!-\!1\big) \!+\! 1\bigg),
\end{align}
the constraint function in (\ref{prob:mix1tr}) is
\begin{align}
f_2(t) = (1\!-\!t) \log_2 \bigg( 1\!+\! \frac{|h_1|^2 P_{\mathrm{max}}}{\sigma^2}\Big( \frac{1}{1-t} \!-\! \frac{P_{\mathrm{SIC}}}{\xi |h_1|^2 P_{\mathrm{max}}}\Big) \bigg), \label{eq:f2}
\end{align}
and
\begin{align}
t_L = 1- \frac{\xi |h_1|^2 P_{\mathrm{max}}}{P_{\mathrm{SIC}} + \xi |h_2|^2 P_{\mathrm{max}}}.
\end{align}
The constraint (\ref{prob:mix1ttrho}) is obtained from $\frac{|h_2|^2}{|h_1|^2} \le 1-\rho \le 1$ and $0 \le t < 1$. It can be proved that $f_1(t)$ and $f_2(t)$ are both concave functions.
\begin{lemma} \label{lemma:concave1}
\textbf{(Concavity of $f_1(t)$)} If ${P_{\mathrm{SIC}}} < {\xi |h_1|^2 P_{\mathrm{max}}}$, the function $f_1(t)$ is concave for all $0 \le t < t_U$. If ${P_{\mathrm{SIC}}} \ge {\xi |h_1|^2 P_{\mathrm{max}}}$, $f_1(t)$ is concave for $1-\frac{\xi |h_1|^2 P_{\mathrm{max}}}{P_{\mathrm{SIC}}} < t < t_U$.
\end{lemma}
\begin{proof}
See Appendix \ref{proof:concave1}.
\end{proof}

\begin{lemma} \label{lemma:concave2}
\textbf{(Concavity of $f_2(t)$)} If $P_{\mathrm{SIC}} \le \xi (|h_1|^2 P_{\mathrm{max}} + \sigma^2)$, the function $f_2(t)$ is concave for $0 \le t < 1$. If $P_{\mathrm{SIC}} > \xi (|h_1|^2 P_{\mathrm{max}} + \sigma^2)$, $f_2(t)$ is concave for $1-\frac{\xi |h_1|^2 P_{\mathrm{max}}}{P_{\mathrm{SIC}} - \xi\sigma^2} < t < 1$.
\end{lemma}
\begin{proof}
By taking the second derivative, we find that
\begin{align}
f_2''(t) = -\frac{1}{(1-t)^3\ln 2}\bigg(\frac{1}{1-t} - \frac{P_{\mathrm{SIC}} - \xi \sigma^2}{\xi |h_1|^2 P_{\mathrm{max}}} \bigg)^{-2} \le 0
\end{align}
holds for both cases. Therefore, the lemma is proved.
\end{proof}

Based on Lemmas \ref{lemma:concave1} and \ref{lemma:concave2}, we have the following conclusion.
\begin{theorem} \label{thm:convex1}
The problem ($\mathbf{P3.1c}$) is a convex optimization problem.
\end{theorem}
\begin{proof}
Since
\begin{align}
1-\frac{\xi |h_1|^2 P_{\mathrm{max}}}{P_{\mathrm{SIC}} - \xi\sigma^2} < 1-\frac{\xi |h_1|^2 P_{\mathrm{max}}}{P_{\mathrm{SIC}}} < t_L, \label{eq:tLB}
\end{align}
$f_1(t)$ and $f_2(t)$ are all concave functions in the range of $t$ defined by (\ref{prob:mix1ttrho}). Hence, the problem ($\mathbf{P3.1c}$) is convex \cite[Chap.~4]{boyd2004convex}.
\end{proof}

As ($\mathbf{P3.1c}$) is a convex optimization problem, there exists a unique optimal solution which either satisfies $f_1'(t) = 0$ or is the boundary point defined by (\ref{prob:mix1tr}) and (\ref{prob:mix1ttrho}). However, the equation $f_1'(t) = 0$ is difficult to be directly solved. Instead, by exploring the concavity of the objective and constraints, the optimal solution can be found by a numerical search algorithm as follows.

Firstly, the feasible set of $t$, denoted by $\mathcal{T}$, is determined according to the constraints (\ref{prob:mix1tr}) and (\ref{prob:mix1ttrho}). Denote $t_1 = \max \{ 0, t_L\}$, $t_2 = t_U$. If $f_2(t_1) \ge r$ and $f_2(t_2) \ge r$, we have $\mathcal{T} = [t_1, t_2]$. If $f_2(t_1) < r$ and $f_2(t_2) \ge r$, due to the concavity of $f_2(t)$, there is a unique $t_3$ that satisfies $f_2(t_3) = r$, which can be found by bisection search \cite[Chap.~2.1]{richard1985douglas}. Then we have $\mathcal{T} = [t_3, t_2]$. If $f_2(t_1) \ge r$ and $f_2(t_2) < r$, similarly we have $\mathcal{T} = [t_1, t_3]$. While if $f_2(t_1) < r$ and $f_2(t_2) < r$, we firstly find any $t_0$ such that $f_2(t_0) \ge r$, and then find $t_3 \in [t_1, t_0]$ and $t_4 \in [t_0, t_2]$ such that $f_2(t_3) = r$ and $f_2(t_4) = r$. The values of $t_0$, $t_3$ and $t_4$ can all be found by bisection search. Thus, the feasible set of $t$ is $\mathcal{T} = [t_3, t_4]$.

Secondly, we search over the feasible set $\mathcal{T} = [t_1, t_2]$ to find the optimal value of $f_1(t)$. If $f_1'(t_1) < 0$, we have $\max_{t \in [t_1, t_2]} f_1(t) = f_1(t_1)$ due to its concavity. Similarly, if $f_1'(t_2) > 0$, we have $\max_{t \in [t_1, t_2]} f_1(t) = f_1(t_2)$. Otherwise, the optimal solution satisfies $f_1'(t) = 0$, which can be found by bisection search as $f_1'(t)$ is monotonically non-increasing. {The algorithm is summarized in Algorithm \ref{alg:num}. The complexity of the algorithm depends on the complexity of bisection search. In particular, the bisection search terminates when a target accuracy $\epsilon>0$ is achieved, i.e., $|f_2(t_3) - r| < \epsilon$ for instance. Thus, the complexity is $\mathcal{O}(\log_2 \frac{1}{\epsilon})$.}

\begin{algorithm}[th]
\caption{Numerical Search Algorithm} \label{alg:num}
\begin{algorithmic}

\REQUIRE The problem ($\mathbf{P3.1c}$)

\ENSURE $\max_t f_1(t)$

\STATE \textbf{1. Determine the feasible set} $\mathcal{T}$

\IF {$f_2(t_1) \ge r$ and $f_2(t_2) \ge r$}

\STATE $\mathcal{T} = [t_1, t_2]$.

\ELSIF {$f_2(t_1) < r$ and $f_2(t_2) \ge r$}

\STATE Find $t_3 \in \{t_1, t_2\}$ so that $f_2(t_3) = r$ by bisection search. Then $\mathcal{T} = [t_3, t_2]$.

\ELSIF {$f_2(t_1) \ge r$ and $f_2(t_2) < r$}

\STATE Find $t_3 \in \{t_1, t_2\}$ so that $f_2(t_3) = r$ by bisection search. Then $\mathcal{T} = [t_1, t_3]$.

\ELSIF {$f_2(t_1) < r$ and $f_2(t_2) < r$}

\STATE Find $t_0 \in \{t_1, t_2\}$ so that $f_2(t_0) \ge r$ by bisection search. Then find $t_3 \in \{t_1, t_0\}$ and $t_4 \in \{t_0, t_2\}$ so that $f_2(t_3) = f_2(t_4) = r$ by bisection search. Finally, $\mathcal{T} = [t_3, t_4]$.

\ENDIF

\STATE \textbf{2. Find the optimal value in set} $\mathcal{T} = [t_1, t_2]$

\IF {$f_1'(t_1) < 0$}

\STATE $\max_t f_1(t) = f_1(t_1)$.

\ELSIF {$f_1'(t_2) > 0$}

\STATE $\max_t f_1(t) = f_1(t_2)$.

\ELSE

\STATE Find $t_0 \in \mathcal{T}$ so that $f_1'(t_0) = 0$ by bisection search. Then $\max_t f_1(t) = f_1(t_0)$.

\ENDIF

\end{algorithmic}
\end{algorithm}

\textbf{Case (1.2):} If $t \ge t_U$ on the other hand, we have $\rho = 0$. Thus, the scheme degrades to time switching, and the rate maximization problem has been solved in Section \ref{sec:ts}. According to Theorem \ref{thm:ts}, the optimal solution is achieved at $t = t_U$, which is also included in the problem ($\mathbf{P3.1c}$). Therefore, we have the following conclusion.

\begin{proposition}
The problem ($\mathbf{P3.1}$) is equivalent to the problem ($\mathbf{P3.1c}$).
\end{proposition}

To this end, in the case that $|h_2|^2 \le (1-\rho) |h_1|^2$, the original problem degrades to the problem ($\mathbf{P3.1c}$). Hence, it can be efficiently solved by Algorithm \ref{alg:num}.

\subsubsection{Subproblem with $|h_2|^2 \ge (1-\rho) |h_1|^2$}
When $|h_2|^2 \ge (1-\rho) |h_1|^2$, the problem can be reformulated as
\begin{subequations} \label{prob:mix2}
\begin{align}
(\mathbf{P3.2}) \; \max_{t, \rho} \;&\; t \log_2 \bigg( 1+ \frac{|h_2|^2 P_{\mathrm{max}}}{\sigma^2}\bigg) + \nonumber\\
&\; (1-t) \log_2 \bigg( 1+ \frac{|h_1|^2 (1-\rho) P_{\mathrm{max}}}{\sigma^2} \bigg) -r, \label{prob:mix2obj}\\
\mathrm{s.t.} \;&\; 1-\rho \le \frac{|h_2|^2}{|h_1|^2}, \label{prob:mix2rho}\\
 \;&\; \textrm{(\ref{prob:mix1r})-(\ref{prob:mix1trho})}. \nonumber
\end{align}
\end{subequations}
As the objective function in (\ref{prob:mix2obj}) is also a decreasing function of $\rho$, it is maximized when $\rho$ achieves its minimum, which is $\rho = \max \big\{1-\frac{|h_2|^2}{|h_1|^2}, \frac{P_{\mathrm{SIC}}}{\xi |h_1|^2 P_{\mathrm{max}}} - \frac{t}{1-t} \big\}$ according to (\ref{prob:mix1Psic}) and (\ref{prob:mix2rho}). Based on the optimal value of $\rho$, there are two cases as follows.

\textbf{Case (2.1):} If $1-\frac{|h_2|^2}{|h_1|^2} \le \frac{P_{\mathrm{SIC}}}{\xi |h_1|^2 P_{\mathrm{max}}} - \frac{t}{1-t}$, i.e., $t \le t_L$, we have $\rho = \frac{P_{\mathrm{SIC}}}{\xi |h_1|^2 P_{\mathrm{max}}} - \frac{t}{1-t}$, and the problem can be reformulated as
\begin{subequations}\label{prob:mix2t}
\begin{align}
(\mathbf{P3.2c}) \quad \max_{t} \;&\; f_3(t) \label{prob:mix2tobj}\\
\mathrm{s.t.} \;&\; f_2(t) \ge r,\label{prob:mix2tf1}\\
{} \;&\; 1- \frac{\xi |h_1|^2 P_{\mathrm{max}}}{P_{\mathrm{SIC}}} < t \le t_L, \label{prob:mix2trho}\\
{} \;&\; 0 \le t < 1, \label{prob:mix2tt}
\end{align}
\end{subequations}
where the objective function is expressed as
\begin{align}
f_3(t) = t \log_2 \bigg( 1+ \frac{|h_2|^2 P_{\mathrm{max}}}{\sigma^2}\bigg) + f_2(t) - r,
\end{align}
$f_2(t)$ is expressed as (\ref{eq:f2}), and the inequality on the left hand side of (\ref{prob:mix2trho}) comes from the constraint $\rho < 1$. Notice that the problem ($\mathbf{P3.2c}$) exists only if $t_L \ge 0$, i.e., $|h_2|^2 \ge |h_1|^2 - \frac{P_{\mathrm{SIC}}}{\xi P_{\mathrm{max}}}$. Similar to the problem ($\mathbf{P3.1c}$), the problem ($\mathbf{P3.2c}$) is also convex as shown in the following theorem.
\begin{theorem} \label{thm:convex2}
The problem ($\mathbf{P3.2c}$) is a convex optimization problem.
\end{theorem}
\begin{proof}
Based on (\ref{eq:tLB}) and Lemma \ref{lemma:concave2}, the function $f_2(t)$ is concave in the feasible range defined by (\ref{prob:mix2trho}) and (\ref{prob:mix2tt}). Hence, the problem ($\mathbf{P3.2c}$) is convex as $f_3(t)$ is the summation of a linear function and a concave function.
\end{proof}
Therefore, the problem ($\mathbf{P3.2c}$) can also be solved by a numerical search algorithm similar to Algorithm \ref{alg:num}.

\textbf{Case (2.2):} If $t \ge t_L$ on the other hand, we have $\rho = 1-\frac{|h_2|^2}{|h_1|^2}$, which is in fact considered and solved in the problem ($\mathbf{P3.1}$).

In summary, the optimal solution for ($\mathbf{P0}$) can be obtained depending on the channel coefficients $h_1, h_2$ and the target rate of UE 1 $R_1 = r$. If $t_L < 0$, i.e., $|h_2|^2 < |h_1|^2 - \frac{P_{\mathrm{SIC}}}{\xi P_{\mathrm{max}}}$, $R_2$ is only constrained by the decoding ability of UE 2, and the problem ($\mathbf{P3.2c}$) is infeasible. Hence, the optimal solution of the original problem ($\mathbf{P0}$) is equivalent to that of the subproblem ($\mathbf{P3.1c}$). If $|h_2|^2 \ge |h_1|^2 - \frac{P_{\mathrm{SIC}}}{\xi P_{\mathrm{max}}}$ on the other hand, the optimal solution for ($\mathbf{P0}$) is the maximum between the solutions for the subproblems ($\mathbf{P3.1c}$) and ($\mathbf{P3.2c}$). The procedure is detailed in Algorithm \ref{alg:gen}.

\begin{algorithm}[th]
\caption{Finding Achievable Rate Region for Generalized Scheme} \label{alg:gen}
\begin{algorithmic}

\REQUIRE $h_1, h_2, R_1 = r$

\ENSURE $R_2$

\IF {$|h_2|^2 < |h_1|^2 - \frac{P_{\mathrm{SIC}}}{\xi P_{\mathrm{max}}}$}

\STATE Find the optimal solution $R_{2,1}$ of subproblem ($\mathbf{P3.1c}$). Then $R_2 = R_{2,1}$.

\ELSE

\STATE Find the optimal solution $R_{2,1}$ of subproblem ($\mathbf{P3.1c}$).

\STATE Find the optimal solution $R_{2,2}$ of subproblem ($\mathbf{P3.2c}$).

\STATE $R_2 = \max\{R_{2,1}, R_{2,2}\}$.

\ENDIF

\end{algorithmic}
\end{algorithm}

It is easily found that the optimal solution is achieved when the decoding power constraint (\ref{prob:mix1Psic}) is satisfied with equality, which is the same with previous two cases. Therefore, the most \emph{economical} way of splitting the power is also optimal for the generalized scheme. Also notice that the condition $|h_2|^2 < |h_1|^2 - \frac{P_{\mathrm{SIC}}}{\xi P_{\mathrm{max}}}$ is consistent with Theorem \ref{thm:ps}. It determines either (\ref{eq:R2_2}) or (\ref{eq:R2_2sic}) is more stringent. It can be seen that (\ref{eq:R2_2sic}) is a critical condition for wireless powered NOMA system. Since part of the received radio signal is split to energy harvester, the signal for information decoding is lowered so that the received SINR for UE 2 in SIC may be lower than that in UE 2. The two cases need to be dealt with separately.

\subsection{Achievable Rate Regions and Discussions} \label{sec:constfig}
Based on Theorems \ref{thm:ts}, \ref{thm:ps} and Algorithm \ref{alg:gen}, the achievable regions of time switching scheme, power splitting scheme and generalized scheme can be depicted. { In particular, the line-of-sight pathloss model $\mathrm{PL} = 30.8 + 24.2 \log_{10} (d)$ (in dB) is adopted \cite{3GPP2010TR}, where $d$ is the transmission distance, and $P_{\max} = 40$ W. The noise power spectral density is -174 dBm/Hz, the bandwidth is 10 MHz, thus we have $\sigma^2 = -104$ dBm. We also set $P_{\mathrm{SIC}} = 80$ mW and $\xi = 0.5$.}

Firstly, the transmit distances for the users are set to $d_1 = 0.5$ m and $d_2 = 10$ m, respectively. In this case, the parameters satisfy $|h_1|^2 > \frac{P_{\mathrm{SIC}}}{\xi P_{\mathrm{max}}}$ and $|h_2|^2 < |h_1|^2 - \frac{P_{\mathrm{SIC}}}{\xi P_{\mathrm{max}}}$. The result is shown in Fig.~\ref{fig:RateNormal}. It can be seen that the proposed generalized scheme achieves a larger rate region compared with the conventional time switching scheme and power splitting scheme. There is a trade-off between the time switching scheme and the power splitting scheme. When $R_1$ is small, the achievable rate of UE 2 with the time switching scheme is larger than that with the power splitting scheme. When $R_2$ is large, the relation reverses. In addition, the \emph{cut-off} line for time switching scheme, $R_1 = \frac{\xi |h_1|^2 P_{\mathrm{max}}}{\xi |h_1|^2 P_{\mathrm{max}} + P_{\mathrm{SIC}}} \log_2 \big( 1+ \frac{|h_1|^2P_{\mathrm{max}}}{\sigma^2}\big)$, is depicted in the figure. {The proposed scheme is also compared with an conventional orthogonal multiple access scheme, TDMA. In TDMA, the BS firstly transmits data to UE 2, while UE 1 harvests energy simultaneously. Then the BS transmits data to UE 1. If the harvested energy is not enough, UE 1 applies power splitting scheme to harvest more energy. It can be seen that the proposed scheme performs better than the conventional TDMA scheme.}

\begin{figure}
\centering
\includegraphics[width=3.4in]{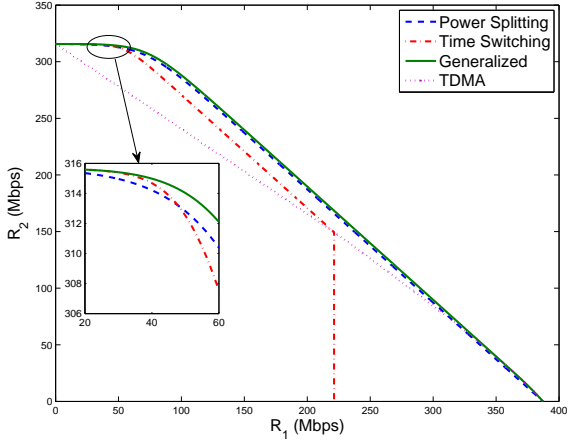}
\caption{Achievable rate regions for constant decoding power ($P_{\mathrm{max}} = 40$ W, $\sigma^2 = -104$ dBm, $P_{\mathrm{SIC}} = 80$ mW, $\xi = 0.5$, $d_1 = 0.5$ m, $d_2 = 10$ m).} \label{fig:RateNormal}
\end{figure}

In Fig.~\ref{fig:RatePSlinear}, we set $d_2 = 1.2$ m so that $|h_2|^2 > |h_1|^2 - \frac{P_{\mathrm{SIC}}}{\xi P_{\mathrm{max}}} > 0$. It can be seen that the boundary curves for the time switching scheme and the generalized scheme overlaps when $R_1 \le \frac{\xi |h_1|^2 P_{\mathrm{max}}}{\xi |h_1|^2 P_{\mathrm{max}} + P_{\mathrm{SIC}}} \log_2 \big( 1+ \frac{|h_1|^2P_{\mathrm{max}}}{\sigma^2}\big)$, which means that $\rho = 0$ is optimal. The reason is that in this condition, the decoding ability of UE 1 becomes the bottleneck, i.e., $R_2$ is constrained by the maximum supportable rate of UE 2's signal decoding in SIC. Hence, all the received power should be split for information decoding in the second sub-slot in order to get a higher SINR. Besides, the boundary curve for the power splitting scheme forms a straight line, which is consistent with Theorem \ref{thm:ps}.

\begin{figure}
\centering
\includegraphics[width=3.4in]{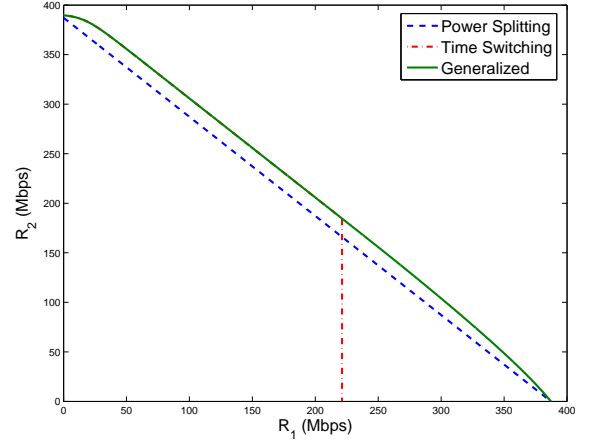}
\caption{Achievable rate regions for constant decoding power ($P_{\mathrm{max}} = 40$ W, $\sigma^2 = -104$ dBm, $P_{\mathrm{SIC}} = 80$ mW, $\xi = 0.5$, $d_1 = 0.5$ m, $d_2 = 1.2$ m).} \label{fig:RatePSlinear}
\end{figure}

Further, the infeasible case for power splitting (i.e., $|h_1|^2 < \frac{P_{\mathrm{SIC}}}{\xi P_{\mathrm{max}}}$) is illustrated in Fig.~\ref{fig:RatePSfail}. In this case, the achievable rate region of the power splitting scheme becomes a straight line with $R_1 = 0$. It can been seen that there is also a \emph{cut-off} line for the achievable rate region of the generalized scheme. Since the pure power splitting scheme is infeasible, the existence of this line comes from the nature of the time switching scheme, i.e., $t$ is strictly larger than 0. The \emph{cut-off} line can be expressed as $R_1 = \max_{0\le t \le 1} f_1(t)$.

\begin{figure}
\centering
\includegraphics[width=3.4in]{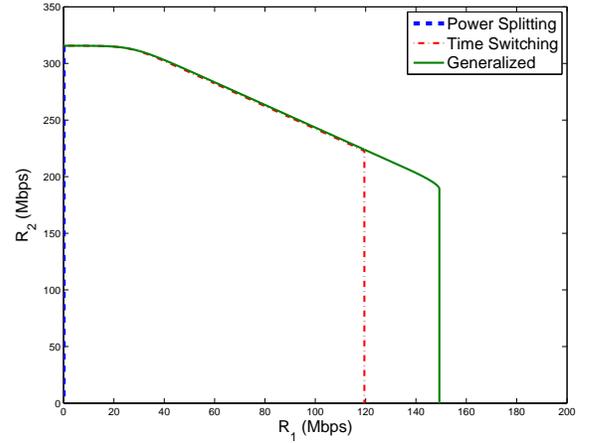}
\caption{Achievable rate regions for constant decoding power ($P_{\mathrm{max}} = 40$ W, $\sigma^2 = -104$ dBm, $P_{\mathrm{SIC}} = 80$ mW, $\xi = 0.5$, $d_1 = 0.75$ m, $d_2 = 10$ m).} \label{fig:RatePSfail}
\end{figure}


\section{Dynamic Decoding Power Consumption Case} \label{sec:dyn}
Recently, the relation between decoding performance and number of iterations (proportional to decoding power consumption) has been found. For instance, an iterative decoding method for LDPC code was studied in \cite[Fig.~8]{kou2001low}. It is shown that to achieve the same error rate, the number of iterations is reduced with the increase of $E_b/N_0$, which corresponds to a decrease of information rate under a fixed SNR. Based on this, dynamic decoding power consumption is considered in this section. It is expected that the achievable rate region can be enlarged compared with dynamic power model. Take the power splitting scheme as an example, the stringent requirement on $h_1$ can be loosen under the dynamic power model, as one can reduce the required decoding power by sacrificing some rate. With this dynamic model, the equality conditions (see Lemma \ref{lemma:equal}) may not hold any more. Consequently, the problem becomes more complex. To solve the problem, both exhaustive search algorithm and low-complex suboptimal algorithm will be proposed.

\subsection{Dynamic Power Model and Problem Formulation}
As shown in \cite[Theorem 1]{blake2015energy}, the decoding power consumption is lower-bounded as
\begin{align}
P_{\mathrm{DEC}} \ge \frac{\omega R}{\sqrt{-\log_2(2p_e)}},
\end{align}
where $\omega$ is a constant parameter related to circuit technology and coding block length, $R$ is the information rate, and $p_e$ is the symbol error rate, which is a function of SINR $\gamma$. Specifically, the symbol error rate for BPSK is expressed as
\begin{align}
p_e = Q(\sqrt{\gamma}) = \int_{\sqrt \gamma}^{+\infty} \frac{1}{\sqrt{2\pi}} e^{-\frac{x^2}{2}} \mathrm{d}x.
\end{align}
{In this paper, the symbol error rate of BPSK is adopted as an approximation. As BPSK usually achieves the lowest error rate, the approximated decoding power consumption is a lower bound. Although there may be a performance gap in practice, the bound is useful to demonstrate the behavior of NOMA system under dynamic power settings.} With this model, the total power consumption of UE 1 is
\begin{align}
P_{\mathrm{SIC}} = \frac{\omega R_1^{(2)}}{\sqrt{-\log_2(2Q(\sqrt{\gamma_1}))}} + \frac{\omega R_2^{(2)}}{\sqrt{-\log_2(2Q(\sqrt{\gamma_2}))}} + P_r, \label{eq:Psicdyn}
\end{align}
where $P_r$ is the constant power consumption of analog receive circuit including filter, low noise amplifier, analog-to-digital converter and so on \cite{cui2005energy}, and
\begin{align}
\gamma_1 &= \frac{|h_1|^2 (1-\rho) P_1^{(2)}}{\sigma^2}, \label{eq:gamma1}\\
\gamma_2 &= \frac{|h_1|^2 (1-\rho) P_2^{(2)}}{|h_1|^2 (1-\rho) P_1^{(2)}+\sigma^2}. \label{eq:gamma2}
\end{align}

In dynamic decoding power case, some constraints are still satisfied with equality as shown in the following lemma.
\begin{lemma} \label{lemma:equal2}
\textbf{(Equality Constraints)} When the maximum of problem ($\mathbf{P0}$) with dynamic power model (\ref{eq:Psicdyn}) is achieved, the constraints (\ref{eq:Pmax_1}), (\ref{eq:R2_1}), and (\ref{eq:Pmax_2}) are satisfied with equality.
\end{lemma}
\begin{proof}
See Appendix \ref{proof:lemma2}.
\end{proof}

Different from Lemma \ref{lemma:equal}, the equality of (\ref{eq:R1_2}) and (\ref{eq:R2_2}) (or (\ref{eq:R2_2sic})) is not guaranteed as the users may sacrifice their data rate to reduce the decoding power. Based on Lemma \ref{lemma:equal2}, ($\mathbf{P0}$) can be reformulated as follows with variables $t, \rho, R_2^{(2)}$ and $P_1^{(2)}$.
\begin{subequations}\label{prob:dyn}
\begin{align}
(\mathbf{P4})  \max_{t, \rho, R_2^{(2)}, P_1^{(2)}} \;&\; t \log_2 \bigg( 1+ \frac{|h_2|^2 P_{\mathrm{max}}}{\sigma^2}\bigg) + (1-t)R_2^{(2)} \label{prob:dynobj}\\
\mathrm{s.t.} \quad\;&\; \log_2(1 + \gamma_1) \ge \frac{r}{1-t}, \label{prob:R1dyn}\\
\;&\; R_2^{(2)} \le \min \{R_{2,1}, R_{2,2} \}, \label{prob:dynR2}\\
\;&\; \xi |h_1|^2 P_{\mathrm{max}}\Big(\frac{t}{1\!-\!t} \!+\! \rho\Big) \ge \nonumber\\
\;&\; \qquad \frac{\omega r}{(1\!-\!t) \sqrt{-\!\log_2(2Q(\sqrt{\gamma_1}))}} \!+\! \nonumber\\
\;&\; \qquad\qquad \frac{\omega R_2^{(2)}}{\sqrt{-\!\log_2(2Q(\sqrt{\gamma_2}))}} \!+\! P_r, \label{prob:dynPsic} \\
\;&\; 0\le t < 1, 0 \le \rho < 1,
\end{align}
\end{subequations}
where $\gamma_1$ and $\gamma_2$ are expressed as (\ref{eq:gamma1}) and (\ref{eq:gamma2}), respectively, $P_2^{(2)} = P_{\mathrm{max}} - P_1^{(2)}$, and
\begin{align}
R_{2,1} &= \log_2 \bigg( 1 + \frac{|h_1|^2 (1-\rho) P_2^{(2)}}{|h_1|^2 (1-\rho) P_1^{(2)}+\sigma^2}\bigg), \label{prob:dynR21} \\
R_{2,2} &= \log_2 \bigg( 1 + \frac{|h_2|^2P_2^{(2)}}{|h_2|^2P_1^{(2)}+\sigma^2}\bigg). \label{prob:dynR22}
\end{align}
Notice that the constraint (\ref{prob:dynPsic}) can be rewritten as
\begin{align}
R_2^{(2)} \le &\; \frac{\sqrt{-\!\log_2(2Q(\sqrt{\gamma_2}))}}{\omega} \bigg( \xi |h_1|^2 P_{\mathrm{max}}\Big(\frac{t}{1\!-\!t} \!+\! \rho\Big) \!-\! \nonumber\\
&\; \quad \frac{\omega r}{(1\!-\!t) \sqrt{-\!\log_2(2Q(\sqrt{\gamma_1}))}} \!-\! P_r \bigg) \nonumber\\
\buildrel def \over = &\; R_{2,\mathrm P}. \label{prob:dynR23}
\end{align}
Thus, the optimal value of $R_2^{(2)}$ can be written as $R_2^{(2)} = \min\{R_{2,\mathrm P}, R_{2,1}, R_{2,2}\}$ which is a function of $t$, $\rho$, and $P_1^{(2)}$, .

Due to the existence of $Q$ function in the power constraint (\ref{prob:dynPsic}), the problem $(\mathbf{P4})$ is quite complex, and well-structured solution is difficult to be found. In the following subsection, we propose optimal and suboptimal algorithms to solve the problem.

\subsection{Optimal and Suboptimal Algorithms}
The optimal algorithm is based on exhaustive search. By exploring the monotonicity of (\ref{prob:R1dyn}), the search range can be restricted. For a given target $r \in \Big[0, \log_2 \big( 1 + \frac{|h_1|^2 P_{\mathrm{max}}}{\sigma^2} \big) \Big]$, we have
\begin{align}
t &\le 1 - \frac{r}{\log_2 \big( 1 + \frac{|h_1|^2 P_{\mathrm{max}}}{\sigma^2} \big)} \buildrel def \over = t_{\mathrm{max}}.
\end{align}
So the search range of $t$ is $t \in [0, t_{\mathrm{max}} ]$. Once both $r$ and $t$ are given, we have
\begin{align}
\rho &\le 1 - \frac{\sigma^2}{|h_1|^2 P_1^{(2)}} \big( 2^{\frac{r}{1-t}} - 1 \big) \nonumber\\
&\le 1 - \frac{\sigma^2}{|h_1|^2 P_{\mathrm{max}}} \big( 2^{\frac{r}{1-t}} - 1 \big) \buildrel def \over = \rho_{\mathrm{max}}.
\end{align}
Hence, the search range of $\rho$ is $\rho \in [0, \rho_{\mathrm{max}} ]$. Finally, we have
\begin{align}
P_1^{(2)} \ge \frac{\sigma^2}{|h_1|^2 (1-\rho)} \big( 2^{\frac{r}{1-t}} - 1 \big) \buildrel def \over = P_{\mathrm{min}}.
\end{align}
Therefore, the search range of $P_1^{(2)}$ is $P_1^{(2)} \in [ P_{\mathrm{min}}, P_{\mathrm{max}} ]$. Since the constraint (\ref{prob:R1dyn}) is automatically satisfied in this range, the feasibility only depends on if $R_{2, \mathrm P}$ is nonnegative.

As $t$, $\rho$ and $P_1^{(2)}$ are continuous variables, to search over the feasible ranges numerically, the feasible regions are discretized by step sizes $\delta t$, $\delta \rho$, and $\delta P$. The selection of the step sizes determines overall search time and accuracy. {The optimality is guaranteed with an acceptable accuracy by sufficiently small granularity of discretization.}

{As the exhaustive search algorithm is time consuming due to its high computational complexity, we further propose a low-complex suboptimal algorithm inspired by the constant decoding power case, where the constraint (\ref{eq:R1_2}) is satisfied with equality. Assume the equality holds in the dynamic decoding power case, then $P_1^{(2)}$ can be represented in terms of $t$ and $\rho$ based on (\ref{prob:R1dyn}), i.e., $P_1^{(2)} = P_{\mathrm{min}}$. As a result, the search over $P_1^{(2)}$ can be omitted, and the complexity is greatly reduced from $O(N^3)$ to $O(N^2)$ where $N$ is the number of iterations for each parameter. The algorithm is summarized as Algorithm \ref{alg:exh}. The operation $\lfloor x \rfloor$ denotes the maximum integer no larger than $x$. It is shown later in the numerical results that the suboptimal algorithm actually achieves the optimal solution in many cases.
}

\begin{algorithm}[th]
\caption{Suboptimal Search Algorithm} \label{alg:exh}
\begin{algorithmic}[1]

\REQUIRE $h_1, h_2, R_1 = r$

\ENSURE $R_2$

\STATE Initialize $R_2 = 0$.

\FORALL {$t = 0, \delta t, 2\delta t, \cdots, \lfloor \frac{t_{\mathrm{max}}}{\delta t} \rfloor \delta t$}

\FORALL {$\rho = 0, \delta \rho, 2\delta \rho, \cdots, \lfloor \frac{\rho_{\mathrm{max}}}{\delta \rho} \rfloor \delta \rho$}

\STATE Set $P_1^{(2)} = P_{\mathrm{min}}$, and calculate $R_{2, \mathrm P}$, $R_{2,1}$, and $R_{2,2}$ according to (\ref{prob:dynR23}), (\ref{prob:dynR21}), and (\ref{prob:dynR22}), respectively.

\IF {$R_{2, \mathrm P} \ge 0$}

\STATE Calculate $R_{2, \mathrm{temp}} = t \log_2 \big( 1+ \frac{|h_2|^2 P_{\mathrm{max}}}{\sigma^2}\big) + (1-t) \min\{ R_{2,1}, R_{2,2}, R_{2,\mathrm P}\}$.

\IF {$R_{2, \mathrm{temp}} > R_2$}

\STATE Update $R_2 = R_{2, \mathrm{temp}}$.
\ENDIF

\ENDIF

\ENDFOR
\ENDFOR

\end{algorithmic}
\end{algorithm}

\subsection{Achievable Rate Regions and Discussions}
In numerical results, the channel model is the same as Section \ref{sec:constfig}. For the dynamic power consumption model, we set $P_r = 30$ mW and $\omega = 0.044$ so that the maximum power consumption equals to 80 mW for fair comparison. {The step sizes for exhaustive search are $\delta t = 10^{-3}$, $\delta \rho = 10^{-3}$ and $\delta P = 0.1$ dB.}

In Fig.~\ref{fig:RateNormaldyn}, the distances of the two users are set as $d_1 = 0.5$ m, $d_2 = 10$ m. It can be seen that the relationship among the curves is similar with that in Fig.~\ref{fig:RateNormal}, {and the suboptimal algorithm performs the same as the exhaustive search one}. However, the achievable rate regions are not convex any more. For instance, there is a inflection point on the curve for the generalized scheme at $R_1 \approx 340$ Mbps. The reason is that in dynamic decoding power case, the problem $(\mathbf{P4})$ may have multiple local optimal points, among which the global optimal one varies for different values of $r$. It is worth noting that a convex rate region can be obtained by time sharing, which is depicted with dotted lines.

\begin{figure}
\centering
\includegraphics[width=3.4in]{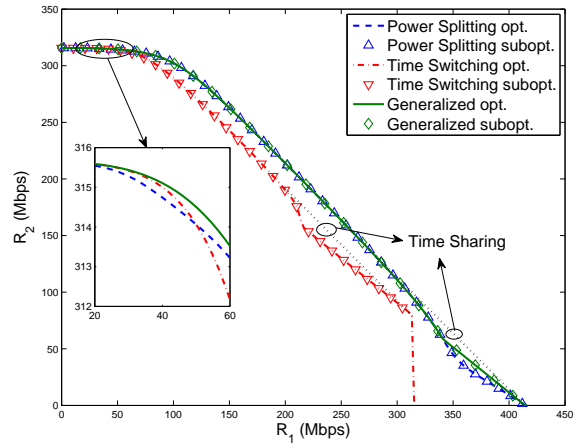}
\caption{Achievable rate regions for dynamic decoding power ($P_{\mathrm{max}} = 40$ W, $\sigma^2 = -104$ dBm, $P_{r} = 30$ mW, $\omega = 0.044$, $\xi = 0.5$, $d_1 = 0.5$ m, $d_2 = 10$ m).} \label{fig:RateNormaldyn}
\end{figure}

%

In Fig.~\ref{fig:RatePSfaildyn}, the distances of the two users are set as $d_1 = 0.75$ m, $d_2 = 10$ m. Different from constant power consumption case where the power splitting scheme is infeasible, a non-zero rate region is achievable in dynamic power consumption case by reducing the decoding power to meet the available harvested energy. Thus, adapting decoding power to data rate helps to enhance the feasibility of the power splitting scheme. {In the power splitting case, the suboptimal algorithm performs worse than the optimal one, and is infeasible when $R_1 < 160$ Mbps. The reason is that to keep (\ref{prob:R1dyn}) satisfied with equality, the decoding power cannot be reduced by increasing SNR of UE 1, as its data rate increases simultaneously.}

\begin{figure}
\centering
\includegraphics[width=3.4in]{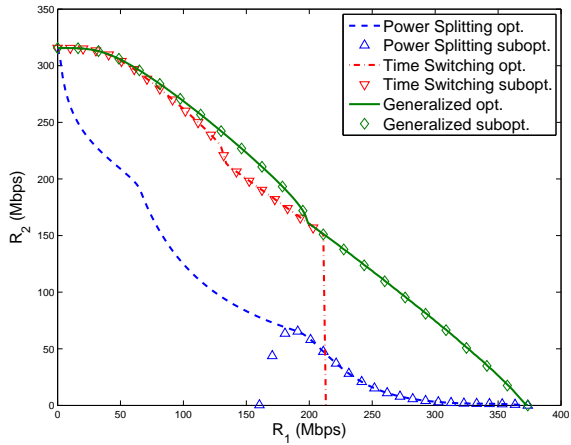}
\caption{Achievable rate regions for dynamic decoding power ($P_{\mathrm{max}} = 40$ W, $\sigma^2 = -104$ dBm, $P_{r} = 30$ mW, $\omega = 0.044$, $\xi = 0.5$, $d_1 = 0.75$ m, $d_2 = 10$ m).} \label{fig:RatePSfaildyn}
\end{figure}

\section{Extended Results and Comparison} \label{sec:comp}
In this section, extended numerical results are presented to show the influence of power consumption and energy harvesting efficiency. We fix the parameters $P_{\mathrm{max}} = 40$ W, $\sigma^2 = -104$ dBm, $\omega = 0.044$, $d_1 = 0.5$ m, $d_2 = 10$ m, and change the values of $P_{\mathrm{SIC}}$, $P_r$ and $\xi$. Notice that when $\omega = 0.044$, the maximum power consumption for SIC decoding is approximately 50 mW in our settings. Thus, we set $P_{\mathrm{SIC}} - P_r = 50$ mW so that the comparison between static power model and dynamic power model is fair.

The performance comparison for the generalized scheme with different power consumption is depicted in Fig.~\ref{fig:MixvsP}. It can be seen that the achievable rate region for dynamic power model is larger than that for static power model. In addition, with the increase of decoding power consumption, the rate regions shrink. It is worth noting that the regions shrink towards $R_1$-axis as only the near user is influenced by the decoding power. Furthermore, only the curve with static power $P_{\mathrm{SIC}} = 100$ mW has a cut-off line since $\xi |h_1|^2 P_{\mathrm{max}} < 100$ mW. It can be predicted that if $P_r \ge 100$ mW, there is also a cut-off line on the curve with dynamic power model.

\begin{figure}
\centering
\includegraphics[width=3.4in]{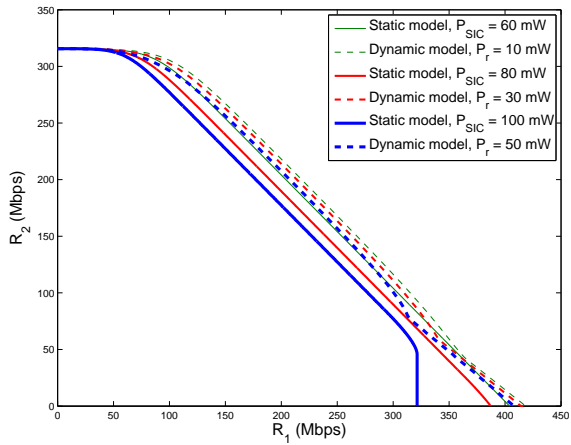}
\caption{Achievable rate regions with generalized scheme versus power consumption ($P_{\mathrm{max}} = 40$ W, $\sigma^2 = -104$ dBm, $\omega = 0.044$, $\xi = 0.5$, $d_1 = 0.5$ m, $d_2 = 10$ m).} \label{fig:MixvsP}
\end{figure}

Fig.~\ref{fig:TSvsP} shows that the achievable rate regions of the time switching scheme are of the same shape, i.e., there is always a cut-off line. In addition, with the increase of power consumption, the difference between the cut-off lines under static power model and dynamic power model becomes small.

\begin{figure}
\centering
\includegraphics[width=3.4in]{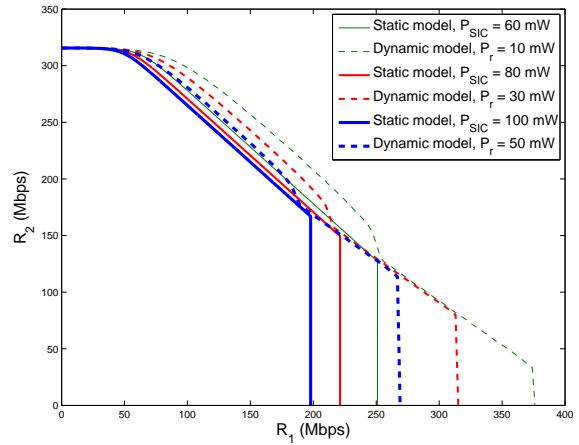}
\caption{Achievable rate regions with time switching scheme versus power consumption ($P_{\mathrm{max}} = 40$ W, $\sigma^2 = -104$ dBm, $\omega = 0.044$, $\xi = 0.5$, $d_1 = 0.5$ m, $d_2 = 10$ m).} \label{fig:TSvsP}
\end{figure}

Then, the performance comparison of the power splitting scheme is shown in Fig.~\ref{fig:PSvsP}. Similarly, the rate region for the dynamic decoding power model is larger than that for the constant decoding power model. In addition, when power splitting scheme is infeasible under constant decoding power model, it still works well under dynamic decoding power model. In the dynamic decoding power model, the performance change over power consumption is gradual.

\begin{figure}
\centering
\includegraphics[width=3.4in]{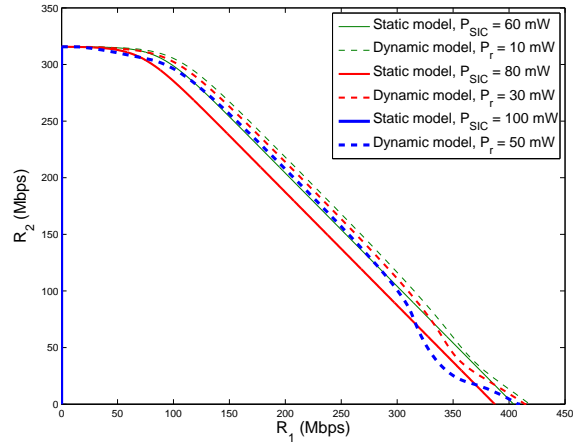}
\caption{Achievable rate regions with power splitting scheme versus power consumption ($P_{\mathrm{max}} = 40$ W, $\sigma^2 = -104$ dBm, $\omega = 0.044$, $\xi = 0.5$, $d_1 = 0.5$ m, $d_2 = 10$ m).} \label{fig:PSvsP}
\end{figure}

{Finally, the influence of energy harvesting efficiency $\xi$ on the rate region is depicted in Fig.~\ref{fig:MixvsXi}. It is shown that the rate regions are enlarged with the increase of the energy harvesting efficiency. The result is similar as Fig.~\ref{fig:MixvsP}, because increasing energy harvesting efficiency has similar impact as decreasing power consumption.
}

\begin{figure}
\centering
\includegraphics[width=3.4in]{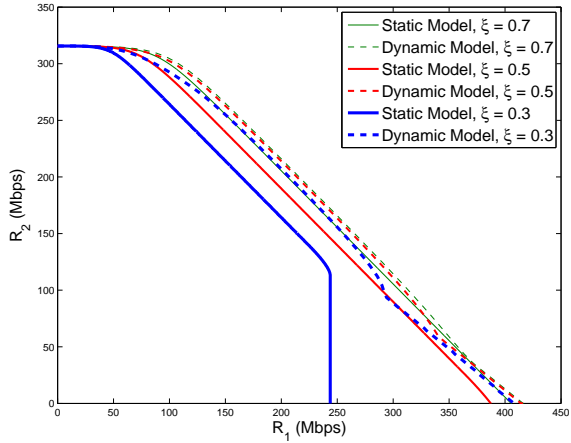}
\caption{Achievable rate regions with generalized scheme versus $\xi$ ($P_{\mathrm{max}} = 40$ W, $\sigma^2 = -104$ dBm, $P_{\mathrm{SIC}} = 80$ mW, $P_{r} = 30$ mW, $\omega = 0.044$, $d_1 = 0.5$ m, $d_2 = 10$ m).} \label{fig:MixvsXi}
\end{figure}

\section{Conclusions} \label{sec:concl}
In this paper, we characterized the achievable rate regions of the NOMA system with wireless powered near user for time switching, power splitting and generalized schemes. Under the constant decoding power model, the achievable rate regions of time switching and power splitting are of closed-form expressions. Specifically, there exists a cut-off line on the boundary of the rate region with the time switching scheme, and the boundary of the power splitting scheme is linear if $|h_2|^2 \ge |h_1|^2 - \frac{P_{\mathrm{SIC}}}{\xi P_{\mathrm{max}}}$. In addition, if $|h_1|^2 \le \frac{P_{\mathrm{SIC}}}{\xi P_{\mathrm{max}}}$, UE 1 can not be self powered by the power splitting scheme. The rate region with the generalized scheme can be derived via solving two convex optimization subproblems, and is shown larger than the conventional ones. Under the dynamic decoding power model, the rate region is further expanded with efficient rate-dependent information decoder. Also, the barrier for applying the power splitting scheme, i.e., the feasibility requirement on $h_1$ is broken. UE 1 can support a lower rate with reduced decoding power consumption.

{ Possible extensions of this work are as follows. Firstly, perfect channel state information is assumed in this paper for theoretical analysis. In practice, it would be interesting to study the influence of channel training and feedback overhead. Secondly, joint transmit power and beamforming design for multi-antenna case would be an interesting research direction.}

\appendices

\section{Proof of Lemma \ref{lemma:equal}} \label{proof:equal}
Since $R_2^{(1)}$ is only constrained by (\ref{eq:R2_1}), and $R_2^{(2)}$ is constrained by both (\ref{eq:R2_2}) and (\ref{eq:R2_2sic}), it is obvious that equality in (\ref{eq:R2_1}) and (\ref{eq:R2_2}) (or (\ref{eq:R2_2sic})) should be satisfied for maximization.

The equality of the power constraints (\ref{eq:Pmax_1}) and (\ref{eq:Pmax_2}) can be proved by contradiction. Take (\ref{eq:Pmax_1}) as an example, assume that $ P_2^{(1)} < P_{\mathrm{max}}$ achieves the maximum average rate. Consider a value $\hat P_2^{(1)}$ that satisfies $ P_2^{(1)} < \hat P_2^{(1)} \le P_{\mathrm{max}}$. It does not violate the decoding power constraint (\ref{eq:Psic}). With $\hat P_2^{(1)}$, a higher rate $R_2^{(1)}$ is achieved according to (\ref{eq:R2_1}), and hence a larger objective. It contradicts the optimality assumption of $ P_2^{(1)}$. The equality of (\ref{eq:Pmax_2}) can be proved similarly.

Given the condition that (\ref{eq:Pmax_1})-(\ref{eq:Pmax_2}) and (\ref{eq:R2_2}) (or (\ref{eq:R2_2sic})) are satisfied with equality, the objective (\ref{prob:mixobj}) is an decreasing function of $P_1^{(2)}$. As $P_1^{(2)}$ is lower bounded by (\ref{eq:R1_2}), the objective is minimized when (\ref{eq:R1_2}) is satisfied with equality.

\section{Proof of Lemma \ref{lemma:mono}} \label{proof:mono}
Denote $\alpha = \log_2 \big( 1+ \frac{|h_2|^2P_{\mathrm{max}}}{\sigma^2}\big), \beta = \frac{|h_2|^2}{|h_1|^2}$. As $0 < |h_2|^2 < |h_1|^2$ by assumption, we have $0 < \beta < 1$. Then $f_0(t)$ can be rewritten as
\begin{align}
f_0(t) = \alpha - (1-t) \log_2 ( \beta2^{\frac{r}{1-t}} + 1-\beta).
\end{align}

The first derivative of $f_0(t)$ at $t = 0$ is
\begin{align}
f_0'(0) = & \Big[ \log_2 ( \beta2^{\frac{r}{1-t}} \!+\! 1\!-\!\beta) \!-\! \frac{\beta r 2^{\frac{r}{1-t}}}{(1-t) ( \beta2^{\frac{r}{1-t}} \!+\! 1\!-\!\beta)} \Big] \bigg|_{t = 0} \nonumber\\
= & \frac{(\beta2^r + 1-\beta) \log_2 ( \beta2^r + 1-\beta) - \beta r 2^r}{ \beta2^r + 1-\beta}.
\end{align}
Denote
\begin{align}
g(\beta) = (\beta2^r + 1-\beta) \log_2 ( \beta2^r + 1-\beta) - \beta r 2^r,
\end{align}
we have
\begin{align}
g''(\beta) = \frac{(2^r-1)^2}{(\beta2^r + 1-\beta) \ln 2} \ge 0,
\end{align}
which indicates that $g(\beta)$ is a convex function of $\beta$. Therefore, the maximum of $g(\beta)$ is achieved at the boundary points, i.e.,
\begin{align}
g(\beta) \le \max \{g(0), g(1) \} = 0.
\end{align}
Hence, we have
\begin{align}
f_0'(0) = \frac{g(\beta)}{ \beta2^r + 1-\beta} \le 0.
\end{align}

Notice that the second derivative of function $f_0(t)$ satisfies
\begin{align}
f_0''(t) = -\frac{ \beta (1-\beta) r^2 2^{\frac{r}{1-t}} \ln{2}}{(1-t)^3(\beta 2^{\frac{r}{1-t}} + 1 - \beta)^2} \le 0,
\end{align}
i.e., the first derivative $f_0'(t)$ is non-increasing. Consequently,
\begin{align}
f_0'(t) \le f_0'(0) \le 0, \; \forall 0 \le t < 1,
\end{align}
which means that $f_0(t)$ is a non-increasing function.

\section{Proof of Lemma \ref{lemma:concave1}} \label{proof:concave1}
The proof can be divided into three steps.

\subsubsection{Simplification}
Denote $\alpha = \log_2 \big( 1+ \frac{|h_2|^2P_{\mathrm{max}}}{\sigma^2}\big), \beta = \frac{|h_2|^2}{|h_1|^2}, \zeta = \frac{P_{\mathrm{SIC}}}{\xi |h_1|^2 P_{\mathrm{max}}}$. The function can be rewritten as
\begin{align}
f_1(t) = \alpha - (1-t) \log_2 \Big( \frac{\beta}{\frac{1}{1-t} - \zeta} \big(2^{\frac{r}{1-t}}-1\big) + 1\Big).
\end{align}
Define
\begin{align}
g(x) = \log_2 \Big( \frac{\beta}{x - \zeta} \big(2^{rx}-1\big) + 1\Big).
\end{align}
As $f_1(t) = \alpha - (1-t)g(\frac{1}{1-t})$, according to the property of perspective of a function \cite[Example~3.20]{boyd2004convex}, to prove that the concavity of $f_1(t)$, we only need to prove the convexity of $g(x)$ for $\max\{1, \zeta\} < x < 1+ \zeta$.

\subsubsection{Convexity of $g(x)$}
To prove the convexity of $g(x)$, we introduce two lemmas as follows.
\begin{lemma} \label{lemma:convex}
Define
\begin{align}
g_0(x) =  \frac{K}{x - \zeta} \big(2^{rx}-1\big) + 1,
\end{align}
where $K > 0, \zeta > 0, r > 0$, and $x > 1$. If $0 < \zeta < 1$, the function $g_0(x)$ is convex for all $1 \le x < 1+ \zeta$. If $\zeta \ge 1$, $g_0(x)$ is convex for $\zeta < x < 1+ \zeta$.
\end{lemma}
\begin{proof}
By taking the second derivative of $g_0(x)$, we have
\begin{align}
g_0''(x) = \frac{K}{(x\!-\!\zeta)^3}\Big(2^{rx}\big((r(x\!-\!\zeta)\ln2\!-\!1)^2+1\big)-2\Big).
\end{align}
Denote
\begin{align}
g_1(x) = 2^{rx}\big((r(x-\zeta)\ln2-1)^2+1\big)-2.
\end{align}
Since
\begin{align}
g_1'(x) = 2^{rx}(r\ln2)^3(x-\zeta)^2 > 0,
\end{align}
$g_1(x)$ is an increasing function. We discuss the sign of $g_0''(x)$ for two cases.

\textbf{(a)}: If $0 < \zeta < 1$, we have $x-\zeta>0$ and
\begin{align}
g_1(x) \ge g_1(1) = 2^{r}\big((r(1-\zeta)\ln2-1)^2+1\big)-2 \stackrel{def}{=} g_2(r).
\end{align}
Again, we have
\begin{align}
g_2'(r) = 2^{r}\ln2\big( (r(1-\zeta)\ln2 -\zeta)^2 + 1-(1-\zeta)^2\big) > 0,
\end{align}
i.e., $g_2(r)$ is an increasing function. Therefore, we have $g_2(r) \ge g_2(0) = 0$, which means that $g_1(x) \ge 0$, and hence $g_0''(x) \ge 0$ for all $x>1$.

\textbf{(b)}: If $x > \zeta \ge 1$, we have $g_1(x) > g_1(\zeta) = 2(2^{r\zeta}-1) > 0$, which again, guarantees that $g_0''(x) \ge 0$. As a result, $g_0(x)$ is convex, and hence, the lemma is proved.
\end{proof}
\begin{lemma} \label{lemma:logconvex}
If the convex positive function $g_0(x)$ is logarithmically convex, i.e., $\log_2 g_0(x)$ is convex, we have $g_0(x) + \delta$ is also logarithmically convex for any $\delta > 0$.
\end{lemma}
\begin{proof}
Since $\log_2 g_0(x)$ is convex, its second derivative is nonnegative. We have
\begin{align}
g_0''(x)g_0(x) - (g_0'(x))^2 \ge 0.
\end{align}
The second derivative of $\log_2 (g_0(x) + \delta)$ is
\begin{align}
g_0''(x)(g_0(x)+\delta) - (g_0'(x))^2 \ge g_0''(x)\delta \ge 0.
\end{align}
Therefore, $\log_2 (g_0(x) + \delta)$ is convex.
\end{proof}
As $g(x) = \log_2(g_0(x) + \frac{K}{\beta} - 1) + \log_2(\frac{\beta}{K})$, based on Lemmas \ref{lemma:convex} and \ref{lemma:logconvex}, $g(x)$ is convex if there exists some $K > \beta$ so that $\log_2g_0(x)$ is convex. The second derivative of $\log_2g_0(x)$ can be expressed as $\frac{g_4(x)}{g_3(x)}$, where the denominator
\begin{align}
g_3(x) = (K2^{rx} + x - \zeta - K)^2 (x-\zeta)^2 \ln2 > 0,
\end{align}
and the nominator
\begin{align}
&g_4(x) = (K2^{rx} + x - \zeta - K)^2 + \nonumber\\
&\; (x\!-\!\zeta)^2 \big( K2^{rx}(r\ln2)^2(x \!-\!\zeta\!-\!K) \!-\! 2K2^{rx}r\ln2 \!-\! 1\big)
\end{align}
can be viewed as a quadratic function of $K$. If the coefficient of $K^2$
\begin{align}
g_5(x) = (2^{rx} - 1)^2 - (x-\zeta)^2 2^{rx} (r\ln2)^2
\end{align}
is positive, as $x$ is bounded, there must exist some $K>\beta$ so that $g_4(x) \ge 0$ holds for all $\zeta < x < 1+\zeta$, which results in $(\log_2g_0(x))'' \ge 0$, and hence, $\log_2g_0(x))$ is convex. Therefore, we only need to prove that $g_5(x) > 0$.

\subsubsection{Proof of $g_5(x) > 0$}
Notice that
\begin{align}
g_5'(x) = \big( 2(2^{rx} \!-\! 1) \!-\! 2(x \!-\! \zeta)r\!\ln\!2 \!-\! ((x \!-\! \zeta)r\ln2)^2 \big) 2^{rx}r\!\ln\!2 .
\end{align}
Denote
\begin{align}
g_6(x) = 2(2^{rx} \!-\! 1) \!-\! 2(x \!-\! \zeta)r\ln2 \!-\! ((x \!-\! \zeta)r\ln2)^2.
\end{align}
As
\begin{align}
g_6''(x) = 2 (2^{rx}-1) (r\ln2)^2 > 0
\end{align}
for all $x > \zeta$, $g_6'(x)$ is increasing, i.e.,
\begin{align}
g_6'(x) = 2(2^{rx} - 1 - (x-\zeta)r\ln2)r\ln2 > g_6'(\zeta) > 0.
\end{align}
As a result, $g_6(x)$ is also an increasing function. We have $g_6(x) > g_6(\zeta) > 0$, which further indicates that $g_5'(x) > 0$. Again, we prove that $g_5(x)$ is increasing and therefore, $g_5(x) > g_5(\zeta) > 0$.

By tracing back from step 3) to step 1), the lemma can be proved.

\section{Proof of Lemma \ref{lemma:equal2}} \label{proof:lemma2}
The equality of (\ref{eq:Pmax_1}) and (\ref{eq:R2_1}) is the same as Lemma \ref{lemma:equal}. We now prove the equality of (\ref{eq:Pmax_2}) by contradiction.

Suppose the strict inequality holds in (\ref{eq:Pmax_2}). Since $Q(\sqrt{\gamma_2})$ is a decreasing function of $\gamma_2$, it is also a decreasing function of $P_2^{(2)}$. Hence, ${\sqrt{-\log_2(2Q(\sqrt{\gamma_2}))}}$ is an increasing function of $P_2^{(2)}$. By properly  increasing $P_2^{(2)}$ and $R_2^{(2)}$ so that the right hand side of (\ref{eq:Psic}) keeps fixed, we can obtain a higher average rate without violating any constraints. It contradicts the optimality assumption, and hence, (\ref{eq:Pmax_2}) must be satisfied with equality.

\bibliographystyle{IEEEtran}
\bibliography{ref}

\begin{IEEEbiography}[{\includegraphics[width=1in,height=1.25in,clip,keepaspectratio]{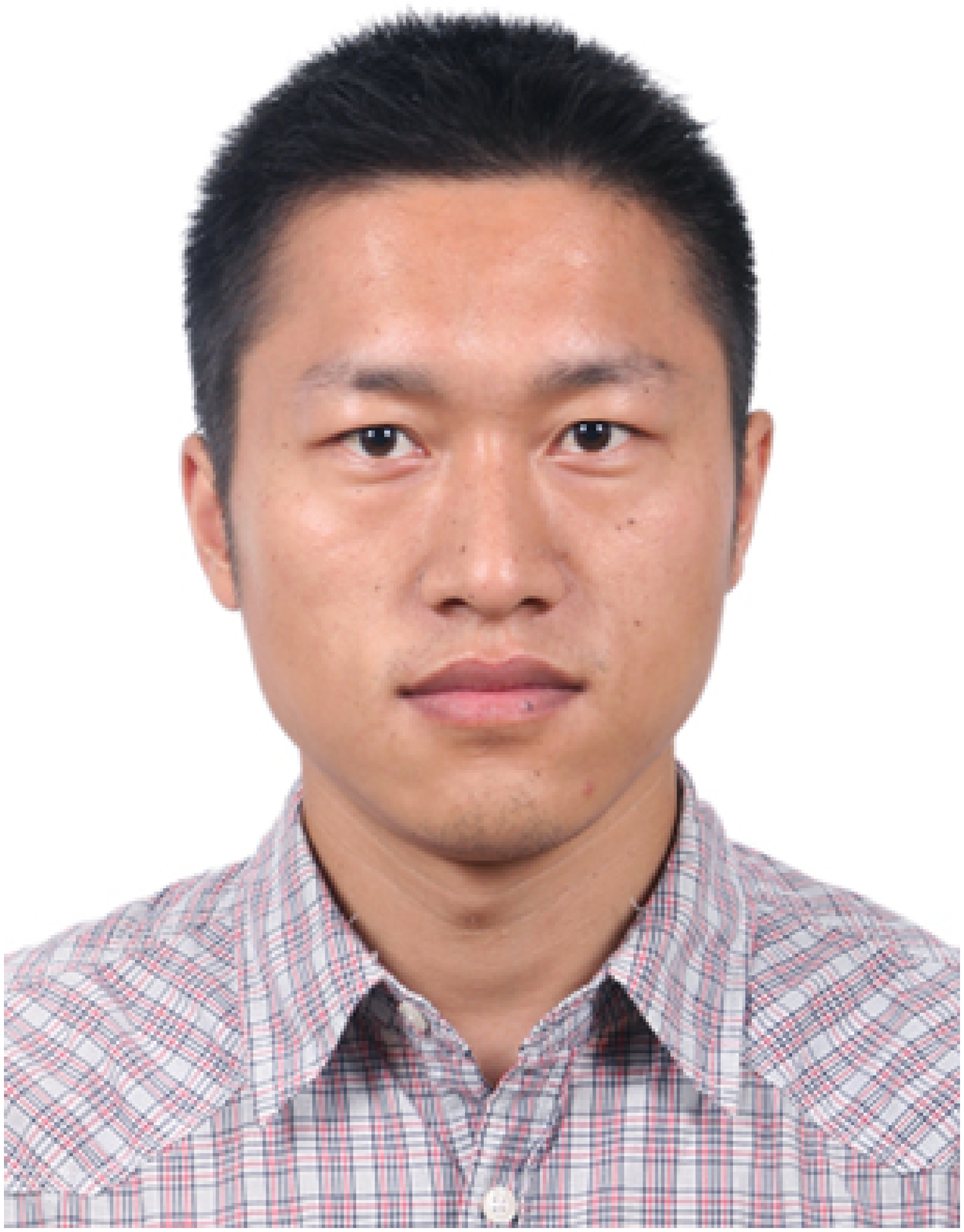}}]{Jie Gong} (S'09, M'13) received his B.S. and Ph.D. degrees in Department of Electronic Engineering in Tsinghua University, Beijing, China, in 2008 and 2013, respectively. From July 2012 to January 2013, he visited Institute of Digital Communications, University of Edinburgh, Edinburgh, UK. During 2013-2015, he worked as a postdoctorial scholar in Department of Electronic Engineering in Tsinghua University, Beijing, China. He is currently an associate research fellow in School of Data and Computer Science, Sun Yat-sen University, Guangzhou, China. He served as workshop co-chair for IEEE ISADS 2015 and TPC member for the IEEE/CIC ICCC 2016/17, the IEEE WCNC 2017, the IEEE Globecom 2017, the IEEE CCNC 2017, and the APCC 2017. He was a co-recipient of the Best Paper Award from IEEE Communications Society Asia-Pacific Board in 2013. He was selected as the IEEE Wireless Communications Letters (WCL) Exemplary Reviewer in 2016. His research interests include Cloud RAN, energy harvesting technology and green wireless communications.
\end{IEEEbiography}

\begin{IEEEbiography}[{\includegraphics[width=1in,height=1.25in,clip,keepaspectratio]{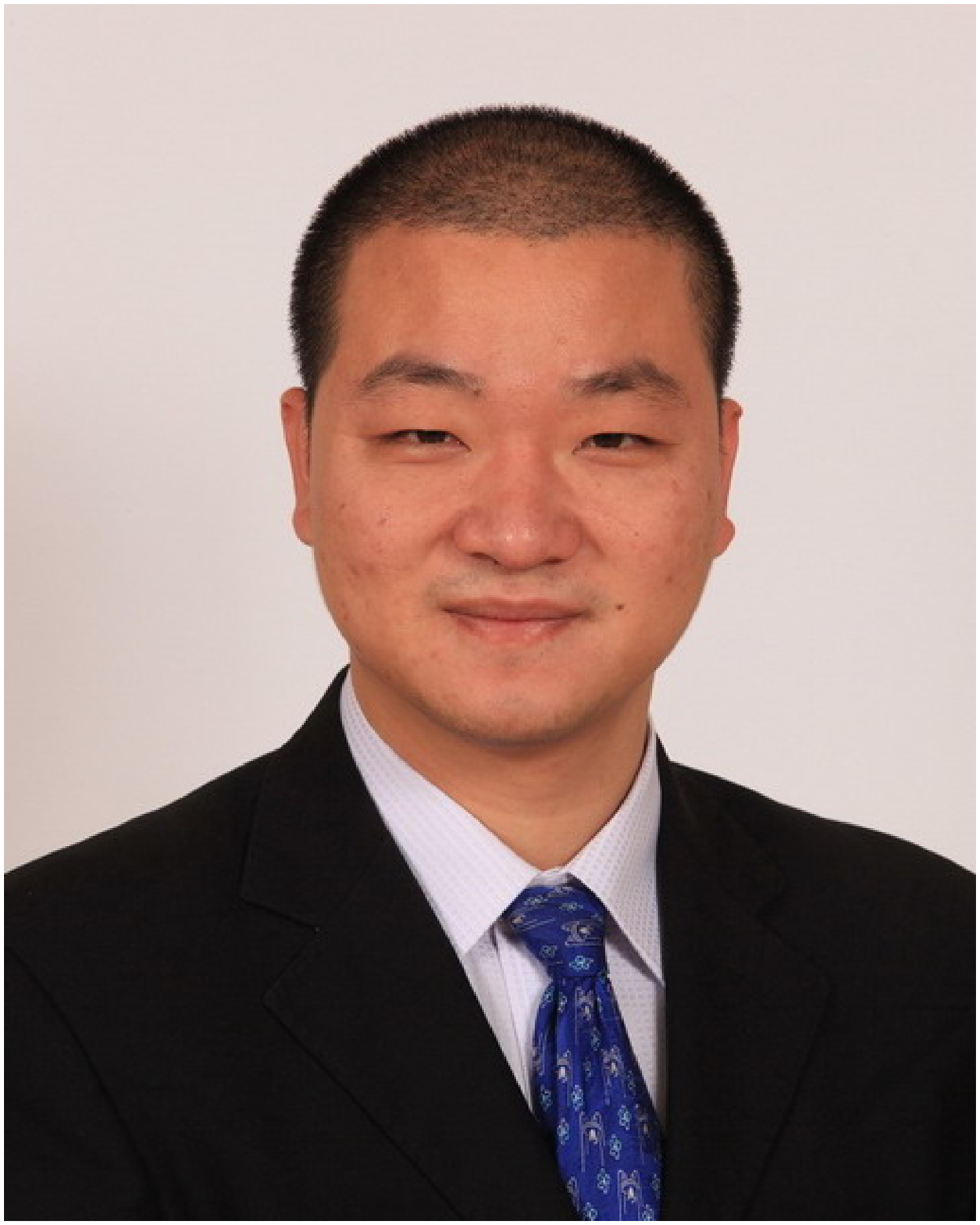}}]{Xiang Chen} (M'08) received the B.E. and Ph.D. degrees both from the Department of Electronic Engineering, Tsinghua University, Beijing, China, in 2002 and 2008, respectively. From July 2008 to July 2012, he was with the Wireless and Mobile Communication Technology R\&D Center, Research Institute of Information Technology in Tsinghua University. From August 2012 to December 2014, was with Tsinghua Space Center, School of Aerospace, Tsinghua University. Since January 2015, he serves as an associate professor at School of Electronics and Information Technology, SYSU-CMU Shunde International Joint Research Institute, Sun Yat-sen University, Guangzhou, China. He is also with Research Institute of Tsinghua University in Shenzhen as a chief researcher (part-time). His research interests mainly focus on 5G wireless communications, Internet of Things, and software radio.
\end{IEEEbiography}


\end{document}